\documentclass[sigconf]{acmart}
\usepackage{epsfig,endnotes}
\usepackage{amsmath, amsthm, amsfonts}
\usepackage{array}
\usepackage{graphicx}
\usepackage{mdframed} 
\definecolor{mygray}{RGB}{211,211,211} 
\usepackage{paralist}

\usepackage{subfiles}
\usepackage{multirow}
\newtheorem{theorem}{Theorem}

\newtheorem{condition}{Condition}

\newcommand{\IP}{\mathbb{P}}
 \usepackage{mathtools}
\usepackage{pifont}
\usepackage[shortlabels]{enumitem}
\usepackage{subfiles}
\usepackage{subcaption}
\usepackage[normalem]{ulem}
\useunder{\uline}{\ul}{}
\usepackage{svg}

\usepackage{xcolor}

\author{De Zhang Lee}
\affiliation{%
  \institution{National University of Singapore}
  \country{Singapore}
  }
\email{dezhanglee@comp.nus.edu.sg}

\author{Han Fang}
\authornote{Corresponding Author}
\affiliation{%
  \institution{National University of Singapore}
  \country{Singapore}
  }
\email{fanghan@nus.edu.sg}

\author{Hanyi Wang}
\affiliation{%
  \institution{Shanghai Jiao Tong University}
  \country{China}
  }
\email{why_820@sjtu.edu.cn}

\author{Ee-Chien Chang}
\affiliation{%
  \institution{National University of Singapore}
  \country{Singapore}
  }
\email{changec@comp.nus.edu.sg}
\copyrightyear{2025}
\acmYear{2025}
\makeatletter

\def\@ACM@copyright@check@cc{}
\makeatother
\setcopyright{cc}
\setcctype{by}
\acmConference[CCS '25]{Proceedings of the 2025 ACM SIGSAC Conference on Computer and Communications Security}{October 13--17, 2025}{Taipei, Taiwan}
\acmBooktitle{Proceedings of the 2025 ACM SIGSAC Conference on Computer and Communications Security (CCS '25), October 13--17, 2025, Taipei, Taiwan}\acmDOI{10.1145/3719027.3765175}
\acmISBN{979-8-4007-1525-9/2025/10}

\begin{document}

\title{Removal Attack and Defense on AI-generated Content Latent-based Watermarking}
\begin{abstract}

Digital watermarks can be embedded into AI-generated content (AIGC) by initializing the generation process with starting points sampled from a secret distribution. When combined with pseudorandom error-correcting codes, such watermarked outputs can remain indistinguishable from unwatermarked objects, while maintaining robustness under whitenoise. In this paper, we go beyond indistinguishability and investigate security under removal attacks. We demonstrate that indistinguishability alone does not necessarily guarantee resistance to adversarial removal. Specifically, we propose a novel attack that exploits boundary information leaked by the locations of watermarked objects. This attack significantly reduces the distortion required to remove watermarks—by up to a factor of $15 \times$ compared to a baseline whitenoise attack under certain settings. To mitigate such attacks, we introduce a defense mechanism that applies a secret transformation to hide the  boundary, and prove that the secret transformation effectively rendering any attacker’s perturbations equivalent to those of a na\"ive whitenoise adversary. Our empirical evaluations, conducted on multiple versions of Stable Diffusion, validate the effectiveness of both the attack and the proposed defense, highlighting the importance of addressing boundary leakage in latent-based watermarking schemes.
\end{abstract}

\maketitle

\section{Introduction}
Recent developments in AIGC (AI-generated content) reignite interests in digital watermarking. Although extensive studies on digital watermarking over the past two decades lead to many robust watermarking schemes, it is still challenging to guard against knowledgeable attackers who are aware of the algorithms and have access to copies of watermarked objects.
Fortunately, the controllable AIGC generation processes provides new surfaces in addressing the challenge. 

Specifically, an AIGC method, for example a Latent Diffusion Model (LDM), can be viewed as a random source parameterized by a condition. Let us write  the generating process as 
${\tt L}(\mathbf{c},\mathbf{s})$
where $\mathbf{c}$ is the condition and $\mathbf{s}$ a random starting point sampled from a multivariate Gaussian distribution in a high dimensional (e.g. $d=4\times 64\times 64)$ real vector space.   Inversion based AIGC watermarking \cite{gunn2025undetectable, DBLP:conf/cvpr/YangZCF0Y24} replaces the original starting point in AIGC with crafted ``starting points'' drawn from a secret distribution.  During detection, given an object, it is first inverted to the starting point and then decides whether the inverted point indeed follows the secret distribution.  
Unlike traditional watermarking approaches, this method on AIGC does not require minimum distortion of a given host to embed the watermark.
Furthermore, by feeding in starting points derived from pseudorandom error correcting (PRC) codes where the codewords are computationally indistinguishable from uniformly sampled binary sequences \cite{DBLP:conf/crypto/ChristG24, gunn2025undetectable}, it is possible to attain indistinguishable (aka undetectable) watermarked objects.

Nonetheless, there are other challenges ahead. In addition to undetectability,  robustness against malicious removal attacks is an important concern.  When given a watermarked object, the attacker wants to slightly perturb it such that the perturbed object  is no longer watermarked. In contrast to the whitenoise attacker who simply adds random noise,  a malicious 
 removal attacker attempt to utilize  knowledge  of the process and observation of other watermarked objects, so as to craft a smaller distortion that would remove the watermark compare to random perturbation.  

Removal attacks can be carried out on the victim at different stages of the generation process. It can be on the final  generated object in the pixel domain \cite{DBLP:conf/nips/ZhaoZSVGKVWL24}, intermediate latent points \cite{DBLP:journals/corr/abs-2412-03283}, or all the way back to the starting point in the latent space.  In this paper, we consider a removal attack that distorts the starting point in the latent space. Such consideration is reasonable.  
{Firstly, the secret is injected into the starting point and thus it would be insightful to understand its robustness against malicious attacks.  Secondly, in practice, the attacker could have more compute resources
or other background information that 
allow for more accurate inversion.
 For security analysis of a defense mechanism, the threat model should assume an attacker with  high capability and able to  invert exactly.  On the other hand, to analyze the attack effectiveness from an attacker's perspective, we should also investigate attackers with various levels of capability. Hence, our  analysis on the defense mechanism consider attackers who can invert exactly (Theorem \ref{thm:mainresult}), and  the empirical studies  considers inversion both with and without errors (Section \ref{section:eval_sd2.1}).}

We formulate a removal attack scenario in which the attackers have access to multiple copies of watermarked objects (referred as {\em known watermark attack}), and compare the attack effectiveness against a nav\"ie whitenoise attacker.
Under the scenario, we propose   an attack targeting indistinguishable watermark,  specifically PRC ~\cite{gunn2025undetectable} and Gaussian  Shading~\cite{DBLP:conf/cvpr/YangZCF0Y24}.
The attack incurs significantly smaller perturbation when compared to whitenoise that attains the same probability of success. Furthermore, the attack is "stealthy", in the sense that the attacked starting point remains indistinguishable from the original (to someone who does not know the seed).
Under reasonable parameters (Figure \ref{fig:bits_flipped} (a)), the perturbation can be reduced by a factor of $15 \times$.   Essentially, this attack attempts to maximize the number of bits flipped in the Hamming space, while incurring small distortion in the real vector space, at the same time ensuring that  changes remain indistinguishable.
In a certain sense, the significant gain is attained by exploiting the fact that an watermarked object's location could leak information about its boundary, in particular, the direction of the nearest non-watermarked object.  Hence, even if the locations are indistinguishable from random points, effective removal could still be possible. 

To defend against the attack, we propose an additional layer of protection by applying a secret, randomly chosen orthonormal transformation to hide the shape of the boundary. While this approach is intuitively appealing, it turns out that certain designs of pseudorandom codes may still leak information. To address this, we introduce a “well-behaved” property, which ensures that all points in the neighborhood of a codeword are also classified as watermarked. With this property, we can  prove that any attack results in distortion indistinguishable from whitenoise. In other words, no attacker can gains an advantage over a na\"ive strategy that simply adds whitenoise.

To validate our analysis, we conduct experiments on Stable Diffusion (SD), considering scenarios where inversion errors occur during both the attack and detection processes. Under this setting (illustrated in Figure \ref{fig:asr_attack_setting}), we observe that our attack achieves a higher watermark removal success rate compared to an adversary limited to adding whitenoise, while maintaining high image fidelity between the watermarked and watermark-removed images.


\subsubsection*{Contributions}
\begin{enumerate}
\item {We propose a threat setting of removal attack in which the adversary has non-adaptive access to a watermarking oracle,  and formulate a method to compare attack effectiveness relative to a baseline whitenoise attacker.}

    \item We present an attack against undetectable watermarks in the latent space of LDMs. The attack exploits the observation that undetectable watermarking schemes might leak information about their detection boundaries. By leveraging this leaked information, the attack significantly reduces the adversarial perturbation required to successfully remove watermarks while remain stealthy.
    \item We propose a novel boundary hiding defense that hides the watermark detection boundary from the adversary. We prove that when the watermarking scheme is pseudorandom and meets a ``well-behaved'' property, and the transformation meets certain conditions, no adversary has an advantage over an adversary limited to adding whitenoise. 
    \item Empirical evaluations on multiple versions of Stable Diffusion demonstrate that our proposed watermark removal attack achieves a high success rate while preserving image quality between watermarked and unwatermarked images. Furthermore, the boundary-hiding defense incurs modest time and memory overhead, making it practical for real-world deployment.
\end{enumerate}

\section{Background and Notations}
\subsection{AIGC}
We focus on the latent diffusion model (LDM) for images generation. A typical LDM for images such as Stable Diffusion \cite{rombach2021highresolution} takes as input a text prompt, some meta-parameters such as number of steps, and a random initial starting point in the latent space of a variational autoencoder (VAE)  like  \cite{DBLP:journals/corr/KingmaW13}\footnote{Usually, the VAE used in the LDM pipeline is a pre-existing model fine-tuned using the data used to train the LDM}.  The generation process first converts the text prompt to a text embedding (e.g. using the CLIP encoder \cite{DBLP:conf/icml/RadfordKHRGASAM21}). Next, the LDM performs forward diffusion conditioned by the text embedding to generate a representation of the generated image in the latent space.  Finally, the VAE decodes the generated latent point into an image in the pixel space. 

We can consider a simplified abstraction.  Let us define the LDM as a deterministic function 
${\tt L} : \mathcal{C} \times \mathcal{L}  \rightarrow \mathcal{L}$ 
that  takes as input a condition $c \in \mathcal{C}$, a starting point $s$ in the VAE's latent space $\mathcal{L} $ and outputs the generated image  in the pixel  space $\mathcal{L}$.  That is, the generated object is:
$$  {\tt L} (c, s).$$
The latent space is a real $d$-dimensional  vector space,  i.e. $\mathcal{L} = \mathbb{R}^d$. The values of $d$ inStable Diffusion 1.5/2.1 and Stable Diffusion XL are $d= 4\times 64\times 64= 16,384$ and $d = 4 \times 128 \times 128 = 65,536$, respectively. Representation of the condition $\mathcal{C}$ is not essential in our analysis.

\subsubsection{Variational Autoencoder}\ \ 
It turns out that for images,  it is more effective to perform diffusion in the latent space, and then ``decode'' it to the pixel space.
The latent space is associated with an encoder $V_{enc}:\mathcal{P} \rightarrow \mathcal{L} $ that map an image in the pixel space to the latent space, and a decoder $V_{dec}: \mathcal{L} \rightarrow \mathcal{P} $ for the inverse.  Overall, the generated image  is 
$$ 
  V_{dec} ( {\tt L} (c,s) ).$$

\subsubsection{Starting point sampling} \ \ 
LDM is desgined with the goal that, when the starting points from $\mathcal{L}$ are sampled from a multivariate Gaussian distribution $\mathcal{N}_d(0,I)$, then the generated objects follow the distribution of the intended distribution specified by the condition (for e.g. a specific prompt describing the appearance of an object) in $\mathcal{P}$.  
Let $Gauss()$ be a probabilistic function that outputs a sample following  the multivariate Gaussian distribution $\mathcal{N}_d(0,I)$.   Hence, we can also express the  process as $$ 
V_{dec}({\tt L} (c, Gauss()))
$$



\subsubsection{Inversion}  
\label{sec:inversion} \ \ The inversion process $inv:  \mathcal{P} \rightarrow \mathcal{L} $ attempts to invert an image  to its starting point. It takes in an image $y$  and outputs a $x$ such that ${\tt L}(c, x) = y $ where $c$ is some condition.  

In practice, inversion incurs noise due to a few factors.  Firstly, the decoder and encoder are not exact inverse of each other. In addition, the condition $c$ might not be known during inversion. The inversion accuracy varies among different adversaries, for example, an adversary might be able to derive an accurate condition $c$ used to generate the image, have more compute resources in minimizing the error, or may have a more accurate inversion algorithm (e.g. \cite{Fang_Chen_Yang_Cui_Zhang_Chang_2025}).

\subsection{Pseudorandom Error Correcting Code}
Pseudorandom Error-Correcting (PRC) codes were introduced by \cite{DBLP:conf/crypto/ChristG24} as a primitive to construction undetectable watermark.
It consists of two components: the encoder and decoder.

\subsubsection*{\bf Encoder}
The encoder $Enc : \{0,1\}^n \times \{0, 1\}^{n_1} \rightarrow \{0, 1\}^{n_2} $ is a probabilistic algorithm, which takes in a $n$-bit key and a $n_1$-bit message
respectively, and outputs a codeword of length $n_2$.    We say that the encoder outputs a pseudorandom codeword if for any $k,m$  the  output  $Enc(k, m)$ is indistinguishable from a truly random sequence for ppt adversaries who do not know $k$.

\subsubsection*{\bf Decoder}
The decoder
$Dec : \{0,1\}^n \times \ \{0,1\}^{n_2} \rightarrow \{0, 1\}^{n_1} \cup \{ \bot \}$,
takes in a $n$-bit key and $n_2$-bit codeword, and outputs the $m$-bit
watermark message. The special symbol  $\bot$ indicates failure in decoding.
Section \ref{sec:indistinguiability} describes how a pseudorandom error correcting code can be incorporated to obtain an undetectable watermark for LDM.

The scheme is error-correcting in the sense that if a codeword is being perturbed by some number of bits not larger than a pre-defined threshold, the decoder will return the correct message.   The key length $n$ is the security parameter and both $n_1$ and $n_2$ are bounded by some polynomials of $n$.



\section{Formulation}

\subsection{LDM Watermarking} \label{subsection:ldm_watermark}

Recap that in a typical generation process, the author chooses a condition $c$, samples a starting point $s$ from a multivariate Gaussian distribution using a sampling function, and then feeds $s$ and $c$ into the forward diffusion process to obtain the generated latent
 $$ {\tt L} (c, Gauss() ).$$

LDM watermarking replaces the sampling $Gauss()$ by a key-sampling function.  As discussed in Section 1, we consider an attacker who can invert and operate on the starting point and thus focus on watermarking of the starting point. A watermarking scheme comprises of two  algorithms:
\begin{itemize}
    \item  $smp:\{0,1\}^{n}  \rightarrow \mathcal{L} $, a probabilistic  keyed-sampling function that generates a starting point.  
\item $det:\{0,1\}^{n} \times \mathcal{L}  \rightarrow \{0,1\}$, a deterministic keyed-detection function that decides whether a starting point is watermarked.
\end{itemize}

There are three phases.
\begin{enumerate}
\item {\em Setup.} The author chooses a private $n$-bit key.  

\item {\em Generation.} The author samples a starting point
$s= smp(k)$ using the key.    Let us call the generated point a watermarked codeword in the latent space.

\item {\em Detection.} Given a starting point $s\in \mathcal{L}$, 
 the author uses the secret $k$ to determine whether it  is watermarked.  The author computes  
                 $$det (k, s).$$
If the output is $1$,  then  $s$ is declared to be watermarked; otherwise, it is non-watermarked.
\end{enumerate}

The parameter  $n$ is the security parameter, and all polynomials referred to are with respect to $n$. 

\subsubsection{False Alarm} \ \ 
False alarm quantifies the chance that a typically generated object is wrongly declared to be watermarked, even though it is not sampled from $smp(k)$.  Given a detector $det$, define false alarm to be the smallest $\gamma$ such that
$$
  \IP [  {det}(k,   Gauss()  ) = 1  ]\leq \gamma,
$$
for any $k$,  where randomness is from the probabilistic $Gauss()$.  



\subsection{Threat Setting}
\label{sec:threatsetting}
We consider an attacker who has access to an oracle that generates watermarking codewords with the same secret key unknown to the attacker.   Let us call this setting a {\em known watermark attack}.  
There are two different attack goals: {\em distinguishability} and {\em removal}.

\subsubsection{Indistinguishability Under Known Watermark Attack}
\label{sec:indistinguiability}
Indistinguishability of watermarked objects was studied in \cite{DBLP:conf/crypto/ChristG24}.   We follow the same definition used in \cite{DBLP:conf/crypto/ChristG24}, with modifications to fit our formulation.

Let us consider a probabilistic polynomial time (ppt) adversary $\mathcal{A}^\mathcal{O}$ who has access to a probabilistic oracle ${\mathcal O}$ that outputs watermarked codewords using the key established during setup.

Specifically,  consider the  following game $IND_{\mathcal A^{\mathcal{O}}}^{smp }$.

\begin{enumerate}
\item A key $k$ of length $n$ is generated and sent to the oracle.
\item Choose a random bit $b_0$.  If $b=0$, generate $s=  Gauss () $; otherwise, generate 
$s =  smp(k) $.
\item  Output 1 if $b=\mathcal{A}^{\mathcal{O}} (1^{n}, y) $; otherwise output 0.
\end{enumerate}

We say that a watermarking scheme  generates indistinguishable watermarks,   if  for any ppt adversary $\mathcal{A}^{\mathcal{O}}$, the difference 
$$
    \left( \IP [ \mbox{IND}_{\mathcal A^{\mathcal{O}}}^{\langle smp, cd\rangle}(n)  = 1] - \frac{1}{2} \right)
$$
is negligible w.r.t. $n$.

\subsubsection*{Remarks}
\begin{enumerate}
    \item Indistinguishability can be realized using  a pseudorandom code $\langle Enc, Dec\rangle $ in this way.  First, designate a message $m_0$ to be ``watermarked''. Next, derive a starting points from   $Enc(k,m_0)$ by treating it as part of the random seed in $Gauss()$. Regardless of subsequent computations on the starting point, the final generated objects in $\mathcal{P}$ will still be indistinguishable.  
    
\item  
   We consider an arithmetic model of computation, that is, a real number occupies a cell and arithmetic operations can be done in a single step.  Since algorithms considered here have constant depth of multiplications (for instance, the proposed secret transformation is a matrix multiplication), polynomial time algorithm under the arithmetic model can be made polynomial time under Turing model on rounded reals.

\end{enumerate}

\begin{figure*}
    \centering
    \includegraphics[width=\linewidth]{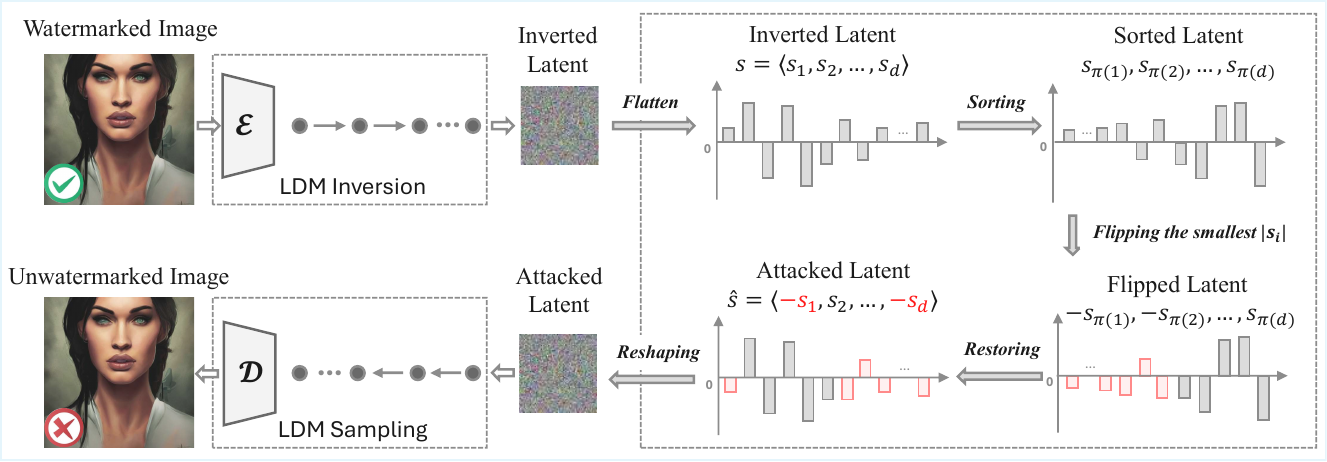}
    \caption{Illustration of the stealthy attack, which targets and flips latent dimensions with the smallest absolute values. This strategy preserves the multivariate Gaussian distribution of the latents while flipping more bits in the watermark codeword compared to the whitenoise attack.}
    \label{figure:attack_flowchart}
\end{figure*}

\subsubsection{Removal Under Known Watermark Attack}
Under a removal attack, the adversary perturbs its victim such that it is not longer watermarked. Similarly, we consider a ppt adversary who has access to the same oracle defined above.  

Let us consider  the following game $\mbox{REM}_{\mathcal{A}^\mathcal{O}}^{ smp, {det} } (n,\epsilon)$.
\begin{enumerate}
\item A key $k$ of length $n$ is generated and sent to the oracle.
\item Generate $s=  smp(k)$.
\item Adversary outputs $\widetilde{s}=\mathcal{A}^\mathcal{O} ( 1^n, s, \epsilon)$.  
\item Let $b=1 $ iff ${det}(k, \widetilde{s}\ )=0$ and $\|s-\widetilde{s}\|_2 \leq \epsilon $.
\item Output $b$.
\end{enumerate}

We say that a ppt $\epsilon$-adversary 
$\mathcal{A}^{\mathcal{O}}$  attains  $\delta$ success rate against the scheme $\langle smp, det\rangle$ if
$\IP [\mbox{REM}_{\mathcal{A}^{\mathcal{O}}}^{smp, {det}} (n, \epsilon) = 1] \geq \delta$ for sufficiently large $n$.

\subsubsection{Whitenoise Adversary}
A whitenoise adversary is a removal attacker who simply adds normalized Gaussian whitenoise to the starting point.  It does not access the oracle and thus in a certain sense, it is a ``dumb'' attack strategy that does not make use of other information. Nevertheless, this attack still succeeds with a sufficiently large noise.  Thus, the whitenoise attacker can serve as a baseline for analysis. 

Let us denote $\mathcal{A}^w_{\tau}$ the whitenoise adversary who introduces  noise with level $\tau$. The $\tau$-whitenoise adversary, when given the targeted $s$, outputs
$$
      \mathcal{A}^w_\tau ( 1^n, s ) = s + \tau \cdot \hat{n} (noise())
$$
where the probabilistic  $noise()$  is sampled from  $\mathcal{N}_d(0,I)$ and 
$\hat{n}(\cdot)$ is the normalization function 
$$ \hat{n} (x) = {x}/{\| x\|_2}, \mbox{\ \ \   for all} \ \ x.$$

We say that an $\epsilon$-adversary $\mathcal{A}^{\mathcal{O}}$ has $\delta$-advantage over whitenoise removal attack   against the scheme $\langle smp, det\rangle$ if, for any $\tau \leq \epsilon$,
 $$
 \IP [\mbox{REM}^{smp}_{\mathcal{A}^{\mathcal{O}} }(n,\epsilon ) =1] -
 \IP [\mbox{REM}^{smp}_{\mathcal{A}^w_{\tau}} (n, \epsilon ) =1] < \delta
 $$
for  sufficiently large $n$.

\subsubsection{Stealthiness}
From a removal attacker's point of view, it is desired to be stealthy in the sense that its output remains indistinguishable w.r.t. a multivariate Gaussian distribution.
Recap that the diffusion process requires the starting point to follow the multivariate Gaussian distribution for good visual quality. Otherwise, the diffusion model may end up generating a lower quality object.  Hence, to preserve quality,  it is desired that the removal is stealthy and does not alter the underlying distribution of the latents.
To illustrate the need for such a requirement, our experiment considers an attack that incurs lower distortion but non-stealthy, and demonstrates that it leads to significant degradation of visual quality. 

\section{Proposed attack} 

This section first describes two known schemes and then demonstrates that there are attacks with a significant advantage against a whitenoise adversary. Two attacks are proposed, a stealthy attack and a non-stealthy attack that minimizes its distortion.  

\subsection{Constructions of $
\langle smp, {det}\rangle$}
Table \ref{table:generation_method} summarizes  constructions by PRC~\cite{gunn2025undetectable} and Gaussian  Shading~\cite{DBLP:conf/cvpr/YangZCF0Y24}.  Both generate binary sequences in  step (1), and then together with additional truly random sources, transform it to multivariate Gaussian sample in step (2).  

PRC attains indistinguishability under known watermark attack, whereas  Gaussian Shading attains a weaker form of indistinguishability.  Gaussian Shading generates watermarked objects that are computationally indistinguishable from the multivariate Gaussian where the distinguisher is only given a single watermarked object, but might fail when the distinguisher sees other copies.


\begin{table*}[]
\centering
\caption{Two SOTA methods (PRC and Gaussian Sharding) that generate watermarked starting points, and each attaining a different form of closeness to the $d$-dimensional multivariate Gaussian $\mathcal{N}_d(0,I)$.}
\label{table:generation_method}
\resizebox{\linewidth}{!}{%
\begin{tabular}{p{0.1\linewidth} | p{0.8\linewidth}}
\hline
Method &
  Description \\ \hline
Gaussian Shading\cite{DBLP:conf/cvpr/YangZCF0Y24} &
  \begin{enumerate}
      \item 
  An $m$-bit binary message is repeated $(d/m)$ times and encrypted via a stream cipher to yield a sequence of  $m$-bit blocks $b_1, \dots,   b_{d/m}$. 
  \item Each block $b_i$ selects the $b_i$-quantile among the $m$-dimensional multivariate  Gaussian distribution, from which the latent value   $x_i$ is sampled via the inverse CDF. 
  \item During detection, given a latent point, determine the quantile of each coefficient and decrypt them. Next, correct the error by using the majority vote.
  \end{enumerate}\\
   \hline
PRC \cite{gunn2025undetectable} &
  \begin{enumerate}
      \item Generate $b_1, \dots, b_d$,  a $d$-bit pseudorandom  codeword for the given message.
       \item Sample $z_1,   \dots, z_d$ independently from $\mathcal{N}(0,1)$ and  set $x_i = (-1)^{1-b_i}|z_i|$. 
       \item During detection, given a latent point, compute the sign of each coefficient and then employ the error-correcting decoder to recover from error.

  \end{enumerate} \\ \hline
\end{tabular}%
}
\end{table*}

\subsection{Removal Attack} \label{section:attack}

Let us describe the removal attack on the PRC watermarking scheme, although the attack can be adapted for Gaussian Shading.

\subsubsection{Stealthy Removal Attack}
The attack takes in two inputs,  a starting point $s=\langle s_1, \ldots, s_d \rangle$ and a real number $\epsilon>0$. Essentially, it flips the maximum possible number of coefficients constrained by the total distortion being smaller than $\epsilon$. Specifically, it performs the following.
\begin{enumerate}
    
        \item Sorts the absolute values of $s$ to find the permutation   $\pi$ 
such that $
   |s_{\pi(1)}| \leq |s_{\pi(2)}| \leq \dots \leq |s_{\pi(d)}|$.
   \item Finds   $  i_0 = \max \{ \ k\  | \ \sum_{i=1}^k 2|s_{\pi(i)}|^2 \leq  \epsilon \}$.
   
        \item 
          Assigns $$
             \hat{s}_i = \begin{cases*}
                    - s_i & if  $ \pi(i) \leq i_0$  \\
                     s_i & Otherwise. 
                 \end{cases*} 
          $$

    \item   Outputs $\hat{s}=\langle \hat{s}_1,\ldots, \hat{s}_d\rangle$
\end{enumerate}
An illustration of this attack is given in Figure \ref{figure:attack_flowchart}.
It can be shown that the attack is stealthy, in the sense that both the initial starting point $s$ and the attacked latent $\hat{s}$ admits the same multivariate Gaussian distribution. Let us omit the proof. 

\subsubsection{Minimal Distortion Attack}  This attack is the same as  the stealthy removal attack, except that it takes in an additional parameter $\gamma$. In this attack, step (3) is replaced by the following:
\begin{enumerate}
\item[(2*)] Finds   $  i_0 = \max \{ \ k\  | \ \sum_{i=1}^k |s_{\pi(i)}|^2 \leq  \epsilon \}$.
\item[(3*)] 
 Assigns $$
             \hat{s}_i = \begin{cases*}
                    - (\gamma/i_0) sign (s_i)  & if  $ \pi(i) \leq i_0$  \\
                     s_i & Otherwise. 
                 \end{cases*} 
$$
\end{enumerate}
where $sign(s_i) = 1$ if $s_i >0$, and -1 otherwise.  Note that in step 2*, we did not multiply $|s_{\pi(i)}|$ by 2 as we did not flip the latents in the minimal distortion attack. 
 
Unlike the Stealthy Removal Attack, the selected coefficients are flipped to the same value $(\gamma/{i_0})$.  However, this modifies the distribution of the latents, and thus the attack is not stealthy and can be easily detectable.  Note that no perturbation can flip more coefficients and yet incur a distortion smaller than $(\epsilon + \gamma)$. In other words, this attack is optimal up to the constant overhead of $\gamma$.

\subsection{Attack Performance}
Under a whitenoise adversary, an individual latent has an equal chance increasing or decreasing. Hence, both proposed attacks will have a higher chance of flipping the signs of the individual latents, as it will always perturb them in the correct direction needed to flip their signs.  Section \ref{section:experiment} gives the empirical results on the effectiveness of the attack.

\subsubsection{Minimal Distortion vs Stealthy Attack.}
While minimal distortion attack outperforms the stealthy removal attack in the number of individual latents flipped, it disrupts the original distribution of the latents. Since many LDMs assume that the starting point is multivariate Gaussian, this adversely impact visual quality, thus rendering the attack unpractical. To illustrate, 
Figure \ref{fig:mindistort_fail}~(c) shows an image generated from the minimal distortion attack, which has poor visual quality compared to the stealthy attack (Figure \ref{fig:mindistort_fail}~(b)). 

\subsubsection{Relationship with Indistinguishability} The successful removal attack does not contradict the fact that the watermarked objects are indistinguishable.  Although the ``location'' of watermarked objects in the latent space is indistinguishable, it is not necessary that the ``shape'' of their boundaries are indistinguishable from perfect spheres.  The attack exploits leaked information on the boundary to find nearby non-watermarked objects. This observation suggests that hiding the boundary is crucial in defending against a removal attack, which leads to our proposed defense method.


\section{Defense: Hiding The Boundary} \label{section:defense}
Removal attacks on indistinguishable watermarked objects are possible because the watermarked latent points do not form a perfect sphere (with the watermark codeword as the center), and the shape of the boundary is revealed to the attackers.  Intuitively, to mitigate the attack, applying a secret random norm-preserving mapping of the generated starting point should hide the shape of the boundary. Such randomness should render  any adversarial perturbation to whitenoise.  Indeed, that is our approach. Surprisingly, it turns out that this is not sufficient as there could be some poorly designed detector that contains a removal backdoor.  Fortunately, if the detector is ``well-behaved'', such vulnerability could be eliminated.
 
This section describes a desired property on the transformation to hide the boundary (Condition 1), a ``well-behaved'' property on the detector  (Section \ref{sec:wellbehaved}), and show the security of the proposed approach under removal attack (Theorem \ref{thm:mainresult}).

\subsection{Sampling With Post-transformation}
\label{section:transformation}
The proposed method introduces an additional step on an existing  scheme  $\langle smp, det\rangle$.   During setup,   a secret
invertible and norm-preserving transformation $t$ on $\mathcal{L}$ is chosen. The transformation is applied after $smp$, and its inverse is applied before $det$.   Let us write the enhanced scheme as $ \langle trn \circ smp,  det \circ itn \rangle$, where
\begin{enumerate}
\item  $(trn \circ smp)(  ( k,t   )) = t (smp(k))$.
\item $(det \circ itn) ( (k, t), y ) =  det (k, t^{-1} (y))$.
\end{enumerate}

We shall show that when the scheme is indistinguishable under known watermark attack and meets a ``well-behaved'' property, and the transformation meets certain conditions, then the enhanced scheme can render any ppt attack to whitenoise.

\subsection{Well-behaved Detector}
\label{sec:wellbehaved}
Before giving the general notion, let us take the PRC as an example. Given a starting point $s\in\mathcal{L}$, the PRC detector first projects $s$ to a binary sequence, and then determines whether its Hamming distance from the nearest codeword is within a threshold.  If so, $s$ is deemed watermarked.  PRC  is ``well-behaved'' in the sense that if given $s$ and its nearest codeword $s_0$,  it is possible to decide whether a point is watermarked without using the secret key.


Let us formulate the above. Consider a scheme $\langle smp, det \rangle$ and a deterministic mapping $map:\mathcal{L} \rightarrow \mathcal {M} $ where $\mathcal{M}$ is some normed space.
We say that the  mapping preserves $det$ if  for all $x_1, x_2$ and any $k$,
$$ (det(k, x_1) =  1)  \wedge  (map(x_1)= map(x_2)) \ \   \Rightarrow  \ \ (det(k, x_2)=1).$$  With respect to a key $k$, let us  call a projected point $y\in \mathcal{M}$ watermarked if one of its pre-image in $\mathcal{L}$ is watermarked,  i.e.  $y$ is watermarked iff $det(k, x) =1$ where $map(x) =y$ for some $x$. Let $W^{\langle det,map\rangle}_k$ be the set of watermarked points in $\mathcal{M}$.  Let   $S_{\gamma} (x) \subseteq \mathcal{M}$ be the sphere centered at $x$ of radius $\gamma$.   

We say that the scheme $\langle smp, det \rangle$ is {\em well-behaved} if there is a ppt $map$ that preserves $det$ and 
a threshold $\gamma$  such that with high probability,   all points in $S_{\gamma}(map(smp()) )$ are watermarked. That is, 
\[
\IP\left[ S_{\gamma}(map(smp(k)) ) \subseteq W^{\langle det,map\rangle}_k    \right]  = 1- neg(n)
\]
where randomness is over  $k$ and the  probabilistic $smp()$.



\subsection{Security Against Removal} \label{subsection:security_T_construction}

Our goal is to show, given a well-behaved indistinguishable scheme, and transformations meeting certain conditions (Condition \ref{condition:transformation}), no ppt adversary can perform better than the whitenoise adversary on the enhanced scheme.  To show that, we first consider an ideal setting when the watermarks are truly random and show that, information-theoretically,  no adversary can have an advantage against a whitenoise adversary. Next, we extend the result to a pseudorandom watermark codeword for the main result (Theorem \ref{thm:mainresult}).

\begin{condition} \label{condition:transformation}
Let us consider $T$,  a random variable of an invertible orthonormal linear transformation on $\mathcal{L}=\mathbb{R}^d$.
\begin{itemize}
\item
    Let $X_1, \dots, X_m$ be i.i.d. random variables of samples in $\mathcal{L}=\mathbb{R}^d$ where each follows distribution $\mathcal{N}(0,I)$
and  let $Y_i = T(X_i)$.  
\item Let $\text{supp}(X_i)$ and $\text{supp}(T)$ denote the support over the probability distribution
    function of the random variable $X_i$ and $T$. Furthermore, 
 let us write ${\bf x}=(x_1, \dots, x_m)$ and ${\bf y}=(y_1, \dots, y_m)$ be instances of the random variables  ${\bf X}=(X_1, \dots, X_m)$ and ${\bf Y}=(Y_1, \dots, Y_m)$. 
\item
Define $S_{\bf y} = \{(t^{-1}(y_1), \dots, t^{-1} (y_m)) : t \in \text{supp}(T)\}$ be the set of all possible values of $(x_1, \dots, x_m)$ given the fixed $ {\bf y}=(y_1, \dots, y_m)$.
\end{itemize}
    We say that   $T$ fulfills  Condition 1 if and only if , 
    \begin{enumerate}[(a)]
          
          \item ($T$ hides its input, except the norm) 
          For any $r, w \in \mathbb{R}^d$ where $||r||_2 = ||w||_2$, the random variables $T(r)$ and $T(w)$ follow the same distribution.
         

        \item  (No leakage of $T$ from multiple observations)    
       For any \ ${\bf y}=(y_1, 
    \ldots, y_m)$,   we have $T  =  (\ T| ({\bf Y}={\bf y})\ )$.  
        
    \end{enumerate}
     
\end{condition}
Condition~\ref{condition:transformation} characterizes the scenario where all instances of $X_1, \dots, X_m$ and $T$ remain equally likely even when $Y_1, \dots, Y_m$ are observed. Specifically, condition 1(a) states that regardless of the perturbation $r$ employed by the adversary, they cannot perform better than by using a randomly sampled perturbation. On the other hand, condition 1(b) ensures that an adversary observing $Y_1, \dots, Y_m$ gains no additional information about the transformation $T$.

\begin{lemma} \label{thm:requirements_of_A}
    Suppose $T$ is uniformly sampled from $f_A$, which is the Haar measure conditioned on orthonormality,  then $T$ satisfies Condition \ref{condition:transformation}. 
\end{lemma}

\begin{proof}


We will first prove Condition 1(a). Since by definition, $T$ is an orthonormal transformation, $T(\cdot)$ is norm-preserving and  spans $\mathbb{R}^d$. That is, for any $r, \tilde{r} \in \mathbb{R}^d$ such that $||r||_2 = ||\tilde{r}||_2$, there exists a $t \in supp(T)$ such that $\tilde{r} = t(r)$. Since $f_A$ is uniform across $supp(T)$, Condition 1(a) is met. 

Next, we prove Condition 1(b).  Since $f_A$ is uniform across all choices of $supp(T)$, all choices across across $S_{\mathbf{y}}$ (as defined in Condition 1) are uniformly likely. Condition 1(b) is met.

   
\end{proof}

Let us define an ideal watermarking scheme $\langle {ideal}_{\tt s}, {ideal}_{\tt d} \rangle$ in the removal game (described in Section 3.2.2). 
On query,  the sampling $ideal_s$ outputs a starting point  $s$ following the distribution $\mathcal{N}_d(0,I)$.  In the game, the ideal detector ${ideal}_{\tt d}$ knows the starting point  $s$. Unlike the original detector, here,
${ideal}_{\tt d}$  has two inputs $s$ and $s'$, where $s$ is the starting point, and $s'$ is the adversary's output. 
Since the ideal detector is well-behave,  there is a ppt algorithm to decide whether $s'$ is watermarked by comparing it with $s$.

Let us further enhance the ideal scheme to the post-transformation scheme $\langle trn \circ ideal_s,  ideal_d\circ itn \rangle$. That is, a secret transformation $t$ is applied after and before $ideal_{\tt s}$ and $ideal_{\tt d}$, respectively.  Now, let us analyze the game $REM^{ideal_{\tt s}, ideal_{\tt d}}_{\mathcal{A}^{\mathcal{O}}} (n, \epsilon)$ in Lemma \ref{lemma:info_theoretic_optimal}, to show that it is not information theoretically possible for an adversary to perform better than a whitenoise adversary.

\begin{lemma} \label{lemma:info_theoretic_optimal}
 Consider the ideal scheme  $\langle trn \circ ideal_{\tt s},  ideal_{\tt d} \circ itn \rangle$.  There does not exist an $\epsilon$-adversary which has $\delta$-advantage over the  whitenoise attacker on $\langle trn \circ ideal_{\tt s},  ideal_{\tt d} \circ itn \rangle$, for all $\epsilon >0$ and $\delta>0$.
\end{lemma}

\begin{proof} 
 Let  $y_1$ the given target that the adversary want to remove its watermark, and $y_2, \ldots, y_m$ be the output of the oracle.  Let $\widetilde{y}$ be the adversary's  output,   $r=  y_1- \widetilde{y}$, and $t$ the secret transformation. Let $Y_i$'s, $T$ and $R$ the respective random variables.
To extract any information on $T$, the adversary can only rely on the input $Y_1, \ldots, Y_m$.  
    Since Condition \ref{condition:transformation}(b) is met, 
$T =  (T|(Y_1, \dots, Y_m))$.  In other words, the adversary gains no information about $T$. Hence $R$ must be independent from $T$.
 

Consider the input to the ideal detector.  The ideal detector receives the original latent point, which is $t^{-1} (y_1)$ for some $t \in supp(T)$, and the inversely transformed  $t^{-1} (\widetilde{y})$.   By linearity of $t$, we have 
 $t^{-1}(\widetilde{y}) = t^{-1}(y_1) + t^{-1}(r)$.   
 
 Consider the situation when the ideal detector faces a whitenoise attacker.  Recap that the whitenoise adversary simply adds normalized whitenoise to $y_1$. Hence, the ideal detector receives the original $t^{-1} (y_1)$ and an inversely transformed whitenoise $t^{-1}( w)$ where $w $ is the whitenoise.

Let us write $R_{\tau}$ be the random variable $(R\ |\   ( \|R\|_2 =\tau)\ ) $, and $W_\tau$ be the noise  introduce by the $\tau$-whitenoise adversary ${\mathcal {A}^w_{\tau}}$. Note that by definition, $T$ is orthonormal and hence it preserves $\ell_2$ norm.

For any $\tau \leq \epsilon$,  by Condition 1(a) and the fact that $R$ is independent from $T$,
$T(R_{\tau})$ admits the same distribution as $T(W_{\tau})$. Since this holds for any $\tau\leq\epsilon$,  the result follows.

    

 
    
\end{proof}

\color{black}

\begin{theorem}
\label{thm:mainresult}
Given  a scheme $\langle smp,det\rangle$
that attains indistinguishable under known watermark attack (Section \ref{sec:indistinguiability}) and is well-behaved (Section \ref{sec:wellbehaved}), and $T$ that meets Condition 1. Consider the post-transformation scheme  $\langle trn\circ smp, det \circ itn \rangle$ with transformation $T$ on $\langle smp, det\rangle$ as defined in  Section \ref{section:transformation}. There does not
    exist a ppt $\epsilon$-adversary which has non-negligible advantage over the $\epsilon$-whitenoise attacker on  $\langle trn\circ smp, det\circ itn \rangle$,  for all  $\epsilon>0$.  
\end{theorem}
\begin{proof}
Suppose not, i.e.,  there exists a ppt $\epsilon$-adversary $\mathcal{A}^\mathcal{O}_0$ who has $\delta$-advantage over the whitenoise attacker $\mathcal{A}^w$ for some $\epsilon$  and  non-negligible $\delta$ in the removal game. 
Let us  construct an ppt attacker  $\mathcal{S}$ 
that distinguishes output of $smp()$ from the truly random $\mathcal{N}_d(0,I)$ by simulating the removal game. 

Given a latent point $s$ which the simulator wants to distinguish, the simulator generates the transformation $t$ and passes $t(s)$ to the adversary. The simulate both $\mathcal{A}^{\mathcal{O}}_0$ and the whitenoise adversary $\mathcal{A}^w$. Let $y_0$ and $y_w$ be the output of $\mathcal{A}^{\mathcal{O}}_0$ and  $\mathcal{A}^w$ respectively.

Note that since the detector is well-behaved,  even without the secret key of the pseudorandom code, the simulator still can carry out detection.  In particular,  there exists a ppt algorithm that can makes the decision based on $(s, t^{-1}(y_0)) $ and $(s, t^{-1}(y_w))$.  Let $b_0$ and $b_w$ be the output of simulated detector on $y_0$ and $y_w$ respectively. 

If $b_0 = b_w$,  then the simulator outputs 0 or 1  with equal probability. If $b_0=0, b_w =1$ (i.e. ${\mathcal A}^\mathcal{O}_0 $ succeeds but not ${\mathcal A}^w$), then outputs 1. If $b_0=1, b_w=0$, outputs 0.   Since, by assumption ${\mathcal A}^\mathcal{O}_0$
 has non-negligible advantage over the whitenoise ${\mathcal A}^w$, the simulator can distinguish  the pseudorandom source $smp()$  which leads to contradiction. 
\end{proof}

\subsection{Well-behaved Detector Counterexample}
To illustrate that the well-behaved condition is necessary in Theorem \ref{thm:mainresult},  let us find a counter-example,  that is,  a scheme $<\overline{smp}, \overline{de}t>$ and $T$ that meets all the the conditions except that it is not well-behaved,  and yet there is an adversary that has advantage over whitenoise. 

Consider an adversary that performs the following:
\begin{enumerate}
  \item  Queries the oracle to get 2 watermarked objects $w_1$ and $w_2$.
  \item  To remove the watermark from a watermarked object $s_0$, the adversary outputs
   $s_0' = s_0 +  \delta_1 (w_1+w_2)$   where $\delta_1$ is a sufficiently small constant.
\end{enumerate}

On the receiving end, let us consider a modified detector based on an existing indistinguishable watermarking scheme. This detector is ``backdoored'' in this way:  on an input $s'$,  performs the following:
\begin{enumerate}
   \item Find the nearest codeword from  $t^{-1}(s')$, say $x'_0$.
   \item Determine whether $t^{-1}(s') - x'_0$ approximates  summation of two codewords, if so, declares $s'$ to be non-watermarked.  If not, it outputs the result when $s'$ is passed into the original detector.
\end{enumerate}

Without inversion noise, the modified detector declares $s_0'$ to be non-watermarked, for all $\delta_1$, and thus we have the counterexample.

\section{Performance Analysis}
\label{section:experiment}
The results in the previous sections assume exact inversion (i.e. the adversary has access to the exact latent starting point used to generate the image), and the analysis is based on $\ell_2$ norm on the latent space. Our  empirical studies consider a typical setup of Stable Diffusion  and analyzes the effect of inversion error and different metrics, including visual similarity, of   the generated images. 


\subsubsection*{Watermarking Baseline}
We evaluate the performance of our proposed watermark removal attacks on the Pseudorandom Codeword (PRC) watermarking scheme, utilizing the code implementation from \cite{DBLP:conf/crypto/ChristG24}. 
The PRC watermarking scheme is parameterized by message length $m$, which is the bit-lenght of messages that can be embedded,   and false positive rate (FPR) $f$, which is the probability that the decoder fails to output a message.  To convert the scheme to a  watermark detector, a message is designated as watermarked and the rest as non-watermarked.
Hence, the false alarm, i.e. the probability $p$ that a random latent vector is detected as watermarked  is $(1-f) \times 2^{-m}$.  To be consistent with the empirical studies in \cite{gunn2025undetectable}, we set the FPR parameter to be $f = 0.01$. Since $(1-f)$ is close to $1$, we can approximate the false alarm as $2^{-m}$.

\subsubsection*{Dataset and Model Used}
We employed Stable Diffusion 2.1 \cite{rombach2021highresolution} with its default parameters as the latent diffusion model, downloaded from HuggingFace\footnote{\url{https://huggingface.co/stabilityai/stable-diffusion-2-1}}. Images were generated using prompts sourced from the Stable Diffusion Prompt (SDP) dataset\footnote{\url{https://huggingface.co/datasets/Gustavosta/Stable-Diffusion-Prompts}},
using the downloaded model's default parameters over 50 timesteps. This is consistent with the evaluations performed in prior work \cite{gunn2025undetectable, DBLP:conf/cvpr/YangZCF0Y24}. In the additional evaluations performed in Appendix \ref{appendix:eval}, we did the same but on the Stable Diffusion 1.5\footnote{\url{https://huggingface.co/stable-diffusion-v1-5/stable-diffusion-v1-5}} and Stable Diffusion XL\footnote{\url{https://huggingface.co/stabilityai/stable-diffusion-xl-base-1.0}} models, both also using their default parameters. 

\subsubsection*{System Configuration}
The experiments were tested on a Linux machine of the following configuration.  Processor: Intel(R) Xeon(R) CPU E5-2620 v4 @ 2.10GHz;
RAM: 256GB;
OS:  Ubuntu 20.04.4 LTS;
GPU: NVIDIA Tesla V100.

\subsubsection*{Evaluation Metrics} 
To evaluate the effectiveness of our proposed attacks, we utilize the \textbf{Attack Success Rate (ASR)} metric, which is the probability that the attacked image is declared non-watermarked.  We also measure    \textbf{proportion of bits flipped} by the attack. The proportion of bits flipped relates to the ASR  since an attack would be successful if the proportion of bits flipped passed some pre-defined threshold. 

To assess the visual similarity between a watermarked image and its distorted version, we employ the following metrics:
\begin{itemize}
    \item \textbf{Root Mean Squared Error (RMSE)}: Measures the overall pixel-wise difference between two images.
    \item \textbf{Structural Similarity Index (SSIM)}: Evaluates the perceived quality and structural similarity between two images.
\end{itemize}

\subsection{Under Assumption of Perfect Inversion}
\label{sec:exactinversion}
In this set of experiments, we consider an attacker and a detector who can perform inversion exactly. 

\subsubsection{Number of Bits Flipped}
Let us analyze the proportion of bits flipped as a function of the $\ell_2$-norm distortion. Figure \ref{fig:bits_flipped} shows the results for both the Stealthy, Maximum Distortion, and whitenoise adversary across different perturbation levels. Additionally, Table~\ref{table:bit_flip_table} provides the ratio of bits flipped by the whitenoise adversary compared to the Stealthy adversary.
As shown in Figure \ref{fig:bits_flipped}, for a given level of distortion (measured by the $\ell_2 $ norm), the stealthy attack flips a significantly larger proportion of bits with a lower magnitude of distortion. 
For example, to flip 5\% of the bits, the Stealthy attack introduces a distortion with $\ell_2$ norm of $\sim3.0$, while the whitenoise attack requires a distortion which has $\ell_2$ norm of $\sim 45.0$ (which is approximately $15\times$ more). 
This aligns with our objective to maximize the number of flipped bits while minimizing the perturbation in the latent space, ensuring that the similarity between the original watermarked image and its corresponding watermark-removed version is preserved.

\paragraph{Minimum Distortion Attack}
The Minimum Distortion attack in Figure \ref{fig:bits_flipped} maximizes the number of bits flipped for a given distortion level. However, we do not apply this watermark removal attack in practice, as it fails to generate high-quality images under Stable Diffusion 2.1. We illustrate this occurrence in Figure~\ref{fig:mindistort_fail}. We attribute this failure to the Minimum Distortion attack altering the underlying distribution of the latents, thereby violating the assumption in Stable Diffusion 2.1 that the latents should follow a multivariate Gaussian distribution.

\begin{figure}
    \centering
    \includegraphics[width=0.75\linewidth ]{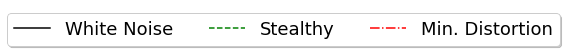}
    
    \begin{subfigure}[b]{0.24\textwidth}
        \centering
        \includegraphics[height=1.2in]{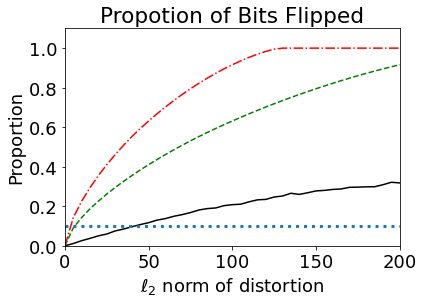}
        \caption{}
    \end{subfigure}%
    \begin{subfigure}[b]{0.24\textwidth}
        \centering
        \includegraphics[height=1.2in]{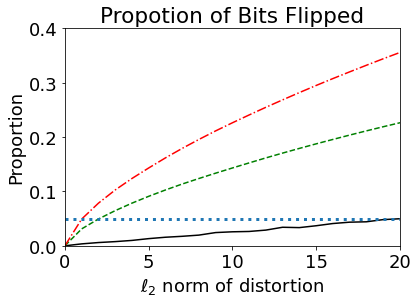}
        \caption{}
    \end{subfigure}%

    \caption{(a) The proportion of bits flipped in the latent space, represented by a vector of $d=4\times64\times64=16,384$ elements sampled from $\mathcal{N}_d(0,1)$, for distortions up to an $\ell_2$ norm of 200. (b) Proportion of bits flipped, focusing on distortions up to an $\ell_2$ norm of 20. The horizontal dotted line on both plots indicate the distortion required by three attack strategies to flip 5\% of the bits.}
    \label{fig:bits_flipped}
\end{figure}
\begin{table}[]
\centering
\caption{The proportion of bits flipped by three attacks (Minimum Distortion, White-Noise, and Stealthy) is shown for varying $\ell_2$-norm distortion levels ($\epsilon$). The rightmost column represents the ratio of bits flipped by the Stealthy adversary compared to the White-Noise adversary at each distortion level.}
\label{table:bit_flip_table}
\resizebox{\columnwidth}{!}{
\begin{tabular}{|c|ccc|c|}
\hline
\multirow{2}{*}{$\epsilon$} &
  \multicolumn{3}{c|}{Attack} &
  \multirow{2}{*}{\begin{tabular}[c]{@{}c@{}}Ratio of\\ W:S\end{tabular}} \\ \cline{2-4}
 &
  \multicolumn{1}{c|}{Min Distortion} &
  \multicolumn{1}{c|}{\begin{tabular}[c]{@{}c@{}}White-Noise\\ (W)\end{tabular}} &
  \begin{tabular}[c]{@{}c@{}}Stealthy\\ (S)\end{tabular} &
   \\ \hline
2.0  & \multicolumn{1}{c|}{0.078} & \multicolumn{1}{c|}{0.005} & 0.049 & 9.80 \\
4.0  & \multicolumn{1}{c|}{0.123} & \multicolumn{1}{c|}{0.009} & 0.078 & 8.67 \\
6.0  & \multicolumn{1}{c|}{0.161} & \multicolumn{1}{c|}{0.015} & 0.102 & 6.80 \\
8.0  & \multicolumn{1}{c|}{0.195} & \multicolumn{1}{c|}{0.019} & 0.123 & 6.47 \\
10.0 & \multicolumn{1}{c|}{0.226} & \multicolumn{1}{c|}{0.025} & 0.143 & 5.72 \\
12.0 & \multicolumn{1}{c|}{0.254} & \multicolumn{1}{c|}{0.028} & 0.161 & 5.75 \\ \hline
\end{tabular}
}

\end{table}

\begin{figure}
    \centering
    \includegraphics[width=0.49\linewidth]{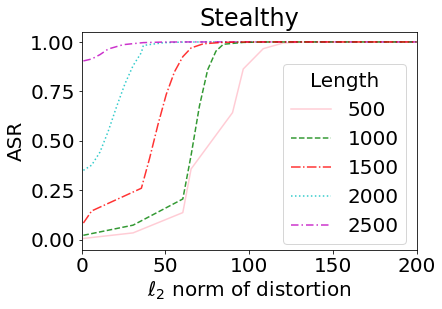}
    \includegraphics[width=0.49\linewidth]{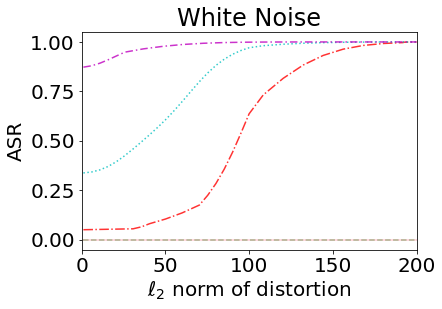}
    \caption{Attack Success Rate (ASR) of the white-noise and stealthy attack across various distortion levels and message lengths.}
    \label{fig:prc_adversary}
\end{figure}

\begin{figure}
    \centering
    \includegraphics[width=0.75\linewidth ]{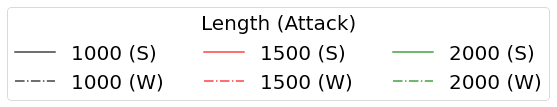}
    
    \begin{subfigure}[b]{0.24\textwidth}
        \centering
        \includegraphics[height=1.15in]{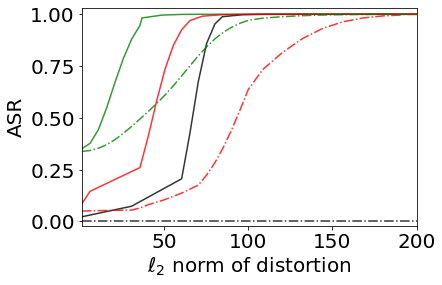}
        \caption{}
    \end{subfigure}%
    \begin{subfigure}[b]{0.24\textwidth}
        \centering
        \includegraphics[height=1.15in]{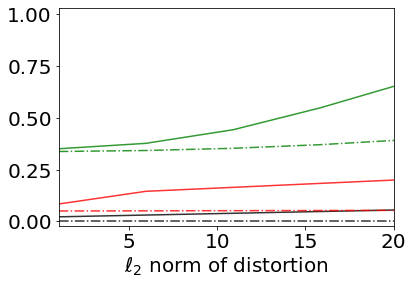}
        \caption{}
    \end{subfigure}%
    \\
    ~

    \caption{(a) Attack Success Rate (ASR) of the White-Noise (W) and Stealthy (S) attackers across different message lengths. (b) ASR plot focusing on distortions up to an $\ell_2$ norm of 20, as distortions exceeding this magnitude are likely to cause significant alterations to the image.}
    \label{fig:prc_asr_compare}
\end{figure}

\begin{figure*}
    \centering
        \begin{subfigure}[b]{0.32\textwidth}
        \centering
        \includegraphics[width=\linewidth]{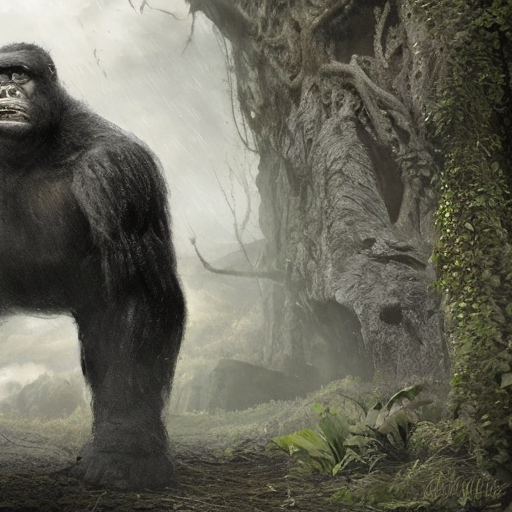}
        \caption{}
    \end{subfigure}%
    \begin{subfigure}[b]{0.32\textwidth}
        \centering
        \includegraphics[width=\linewidth]{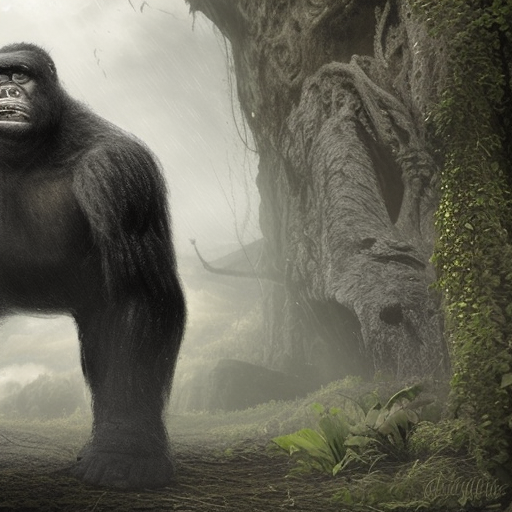}
        \caption{}
    \end{subfigure}%
    \begin{subfigure}[b]{0.32\textwidth}
        \centering
        \includegraphics[width=\linewidth]{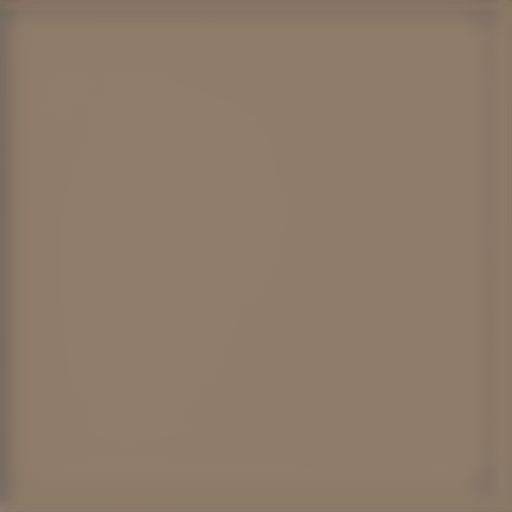}
        \caption{}
    \end{subfigure}%
    \caption{(a) The original watermarked image. (b) Outcome of the stealthy attack with an $\ell_2$-norm distortion of $8$ on the starting point, resulting in approximately $12\%$ bit flip. (c) Outcome of the minimum distortion attack with an $\ell_2$-norm distortion of $8$ on the starting point, achieving approximately $16\%$ bit flip.  The latent produced by the minimum distortion attack fails to produce a high-quality image under Stable Diffusion 2.1.}
    \label{fig:mindistort_fail}
\end{figure*}

\subsubsection{Distortion Required to Remove Watermark}

The impact of perturbation on watermark detection is depicted in Figure~\ref{fig:prc_adversary} for watermark message lengths of 500, 1,000, 1,500, and 2,000 bits. As expected, shorter watermark message lengths exhibit greater resilience to perturbations compared to longer ones. Moreover, for message lengths of 500 and 1,000 bits, the whitenoise adversary fails to remove the watermark even when the perturbation reaches an $\ell_2$ norm of 200. However, the stealthy attack proves significantly more effective, requiring considerably less perturbation (compared to the whitenoise adversary) to erase the watermark across all message lengths, as illustrated in Figure \ref{fig:prc_asr_compare}.

\subsubsection{Robustness under Secret Transformation}
Theorem \ref{thm:mainresult} dictates that there is no noticeable difference between the whitenoise attacker and the Stealthy attack on the enhanced scheme.   For double checking, we perform the experiments in Section \ref{sec:exactinversion} on the enhanced scheme, the results for Stealthy and whitenoise are within statistical errors.



\subsection{Evaluation on Stable Diffusion 2.1} \label{section:eval_sd2.1}


In this section, we evaluate our proposed attacks and defenses on Stable Diffusion 2.1 (the same model used in the evaluation of \cite{gunn2025undetectable}). Further evaluations on Stable Diffusion 1.5 and Stable Diffusion XL are provided in Appendix~\ref{appendix:eval}.
We consider three attacker capabilities (\textbf{AC}) that differ in their inversion accuracy.
\begin{enumerate}
    \item {\bf AC1} : {\em Exact inversion}. 
    \begin{enumerate}
        \item Given a watermarked image $I$, the adversary inverts without noise to get the  starting point $s$. 
        \item The adversary performs attack on $s$ to get $s'$. \item The adversary performs diffusion on  $s'$ with the original prompt and obtains the final image $I'$.
    \end{enumerate}
    
    \item {\bf AC2}: {\em DDIM inversion with original prompt}. 
        \begin{enumerate}
        \item Given a watermarked image $I$, the adversary inverts $I$ using DDIM with the original prompt to get $s'$.
        \item Adversary performs attack on $s$ to get $s'$.
        \item The adversary performs diffusion on $s'$ with the original prompt to obtain the final image $I'$. 
        \end{enumerate}
    \item {\bf AC3}:  {\em DDIM inversion with null prompt}. 
        \begin{enumerate}
          \item Given a watermarked image $I$, the adversary inverts $I$ using DDIM with the null  prompt to get $s'$.
        \item Adversary performs attack on $s$ to get $s'$.
        \item The adversary performs diffusion on $s'$ with the null prompt to obtain the final image $I'$. 
    
    \end{enumerate}
\end{enumerate}

We employ the PRC watermark detector implementation by \cite{gunn2025undetectable}, which is as follows:
 Given a watermarked image $I$, the detector performs DDIM inversion (conditional on the null prompt) to obtain the approximate latent $s$. Watermark detection is then done on $s$, 
 

We performed our evaluation on the (1) {\em whitenoise}, and (2) {\em Stealthy} attacks.
We left out the {\em Minimum Distortion} attack from our evaluation in this section, as the generated images are of low quality (see for example Figure \ref{fig:mindistort_fail}). 

The following metrics were used to evaluate the success rates of watermark removal and the similarity between watermarked and watermark-removed images.
\begin{enumerate}
\item  Attack Success Rate (ASR)
\item Structural Similarity Index Measure (SSIM)
\item  Root Mean Squared Error (RMSE) 
\end{enumerate}

With the proposed boundary-hiding transformation defense, the distortion introduced by the stealthy attack effectively becomes equivalent to that of the whitenoise attack after the transformation is inverted prior to watermark detection. Consequently, when the {\em Stealthy} attack is applied to the defended method, its effectiveness matches that of the {\em whitenoise} attack on the undefended method.

In total, we evaluate 6 different setups, varying along two dimensions: adversarial strategy and the inversion method. For each setup, 100 prompts were randomly selected from the prompt dataset and used to generate 100 images. The reported metrics (ASR, RMSE, SSIM) represent the average values computed across these 100 images per setup.

\begin{table*}[]
\centering
\caption{Comparison of our proposed attack versus the white-noise attack under three inversion settings—perfect inversion, DDIM conditioned on a null prompt, and DDIM conditioned on the same prompt—across varying perturbation levels (measured by the $\ell_2$ norm, on top of the inversion error introduced). We report Attack Success Rate (ASR) for watermark removal, as well as image similarity metrics: Root Mean Squared Error (RMSE) and Structural Similarity Index Measure (SSIM) between the original and attacked images, rounded to 4 and 2 significant figures, respectively. In all inversion settings, our proposed attack (\underline{underlined}) outperforms the white-noise baseline. For each perturbation level, the highest ASR is highlighted in \textbf{bold}.}
\label{table:eval_sd_attack}
\resizebox{\linewidth}{!}{%
\begin{tabular}{|l|l|cccccccccccrccc|}
\hline
\multirow{3}{*}{\begin{tabular}[c]{@{}c@{}}Attacker\\ Capability\end{tabular}} &
  \multirow{3}{*}{Attack} &
  \multicolumn{15}{c|}{Additional Perturbation, $\epsilon$ (Excluding Inversion Error)} \\ \cline{3-17} 
 &
   &
  \multicolumn{3}{c|}{4.0} &
  \multicolumn{3}{c|}{6.0} &
  \multicolumn{3}{c|}{8.0} &
  \multicolumn{3}{c|}{10.0} &
  \multicolumn{3}{c|}{12.0} \\ \cline{3-17} 
 &
   &
  ASR &
  RMSE &
  \multicolumn{1}{c|}{SSIM} &
  ASR &
  RMSE &
  \multicolumn{1}{c|}{SSIM} &
  ASR &
  RMSE &
  \multicolumn{1}{c|}{SSIM} &
  ASR &
  RMSE &
  \multicolumn{1}{c|}{SSIM} &
  ASR &
  RMSE &
  SSIM \\ \hline
\multirow{2}{*}{AC1} &
  White Noise &
  0.00 &
  0.0001 &
  \multicolumn{1}{c|}{0.96} &
  0.00 &
  0.0004 &
  \multicolumn{1}{c|}{0.93} &
  0.03 &
  0.0009 &
  \multicolumn{1}{c|}{0.92} &
  0.10 &
  0.0070 &
  \multicolumn{1}{r|}{0.75} &
  0.14 &
  0.0120 &
  0.71 \\
 &
  {\ul Stealthy} &
  {\ul \textbf{0.06}} &
  0.0001 &
  \multicolumn{1}{c|}{0.96} &
  0.06 &
  0.0005 &
  \multicolumn{1}{c|}{0.93} &
  {\ul 0.12} &
  0.0007 &
  \multicolumn{1}{c|}{0.91} &
  0.14 &
  0.0062 &
  \multicolumn{1}{r|}{0.85} &
  {\ul 0.17} &
  0.0091 &
  0.76 \\ \hline
\multirow{2}{*}{AC2} &
  White Noise &
  0.00 &
  0.0002 &
  \multicolumn{1}{c|}{0.97} &
  0.01 &
  0.0004 &
  \multicolumn{1}{c|}{0.95} &
  0.04 &
  0.0014 &
  \multicolumn{1}{c|}{0.93} &
  0.14 &
  0.0081 &
  \multicolumn{1}{r|}{0.73} &
  0.17 &
  0.0106 &
  0.72 \\
 &
  {\ul Stealthy} &
  {\ul \textbf{0.06}} &
  0.0001 &
  \multicolumn{1}{c|}{0.96} &
  {\ul \textbf{0.09}} &
  0.0004 &
  \multicolumn{1}{c|}{0.94} &
  {\ul 0.25} &
  0.0007 &
  \multicolumn{1}{c|}{0.92} &
  {\ul 0.38} &
  0.0079 &
  \multicolumn{1}{r|}{0.86} &
  {\ul 0.43} &
  0.0069 &
  0.81 \\ \hline
\multirow{2}{*}{AC3}  &
  White Noise &
  0.00 &
  0.0001 &
  \multicolumn{1}{c|}{0.97} &
  0.00 &
  0.0003 &
  \multicolumn{1}{c|}{0.94} &
  0.03 &
  0.0009 &
  \multicolumn{1}{c|}{0.91} &
  0.14 &
  0.0081 &
  \multicolumn{1}{r|}{0.76} &
  0.14 &
  0.0194 &
  0.73 \\
 &
  {\ul Stealthy} &
  {\ul 0.04} &
  0.0001 &
  \multicolumn{1}{c|}{0.96} &
  {\ul \textbf{0.09}} &
  0.0004 &
  \multicolumn{1}{c|}{0.95} &
  {\ul \textbf{0.31}} &
  0.0008 &
  \multicolumn{1}{c|}{0.92} &
  {\ul \textbf{0.39}} &
  0.0058 &
  \multicolumn{1}{r|}{0.87} &
  {\ul \textbf{0.43}} &
  0.0076 &
  0.83 \\ \hline
\end{tabular}%
}
\end{table*}

\begin{figure}
    \centering
    \includesvg[width=\linewidth ]{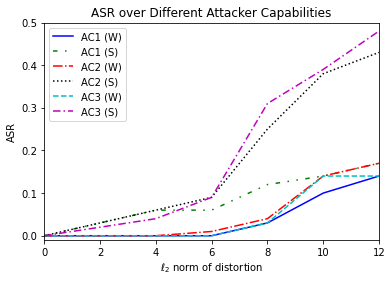}

    \caption{The Attack Success Rate (ASR) for the various Attacker Capability (AC) settings is compared for the White-Noise (W) and Stealthy (S) adversaries. We observe that while the Stealthy adversary consistently outperforms the White-Noise adversary across all attacker capability settings, this advantage is significantly more pronounced under AC2 and AC3 (i.e., when DDIM inversion is applied to approximate the latent starting point).}
    \label{fig:asr_attack_setting}
\end{figure}

\subsubsection{Without Boundary-Hiding Defense}

The results of our evaluation are presented in Table~\ref{table:eval_sd_attack}. Across all the perturbation levels considered, our method consistently outperforms the whitenoise baseline, achieving higher ASR while maintaining greater image fidelity. In particular, AC3 successfully mantains high image fidelity For illustration, we present an example of an image subjected to our attack under varying levels of perturbation, along with the corresponding RMSE and SSIM values in Figure~\ref{fig:diff_ssim_mse}. Finally, in Figure \ref{fig:example_removed}, we present examples of watermarked images and their corresponding watermark-removed versions, illustrating the high perceptual similarity between them.

\paragraph{Improved Attack Performance in AC2 and AC3}
 Our proposed watermark removal attack performs better on DDIM-inverted latents, both when conditioned on a null prompt and when conditioned on the original prompt used to generate the image. This is also illustrated in  Table~\ref{table:eval_sd_attack}, which shows that the ASR of the Stealthy adversary under AC2 and AC3 (i.e. when DDIM inversion is applied to approximate the latent starting point) is significantly higher than the ASR of the whitenoise adversary. This is because AC2 and AC3 attacks the latent space obtained via inversion, which aligns with the method used by the watermark detector to extract the watermarking codeword. In particular, the detector performs null-text inversion on the image to extract the codeword (which aligns with more closely with AC3). 
However, null-text inversion has been shown to introduce mild background inconsistencies \cite{mokady2023null, jiang2025pixelman}, which is seen by the images generated under AC3. Nonetheless, our evaluations suggest these inconsistencies are minor, and the resulting image pairs remain perceptually similar, as shown in Figures \ref{fig:example_removed} and \ref{fig:example_additional_sd21}.

\begin{figure*}[!htbp]
    \centering
    \begin{subfigure}[b]{0.2\linewidth}
        \centering
        \includegraphics[height=1.2in]{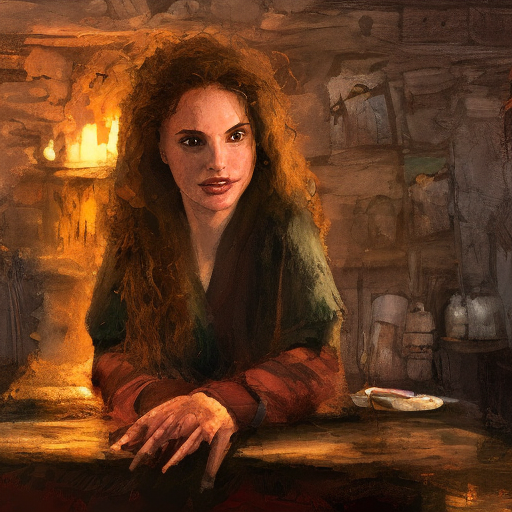}
        \caption{Original}
    \end{subfigure}%
    \begin{subfigure}[b]{0.2\linewidth}
        \centering
        \includegraphics[height=1.2in]{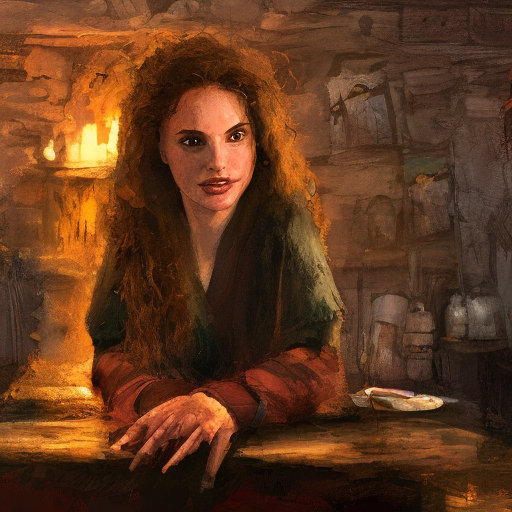}
        \caption{(4.0, 0.9579, 0.00007)}
    \end{subfigure}%
    \begin{subfigure}[b]{0.2\linewidth}
        \centering
        \includegraphics[height=1.2in]{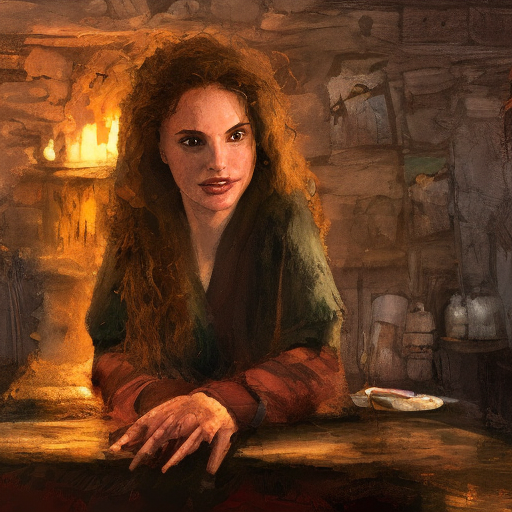}
        \caption{(6.0, 0.9501, 0.00033)}
    \end{subfigure}%
    \begin{subfigure}[b]{0.2\linewidth}
        \centering
        \includegraphics[height=1.2in]{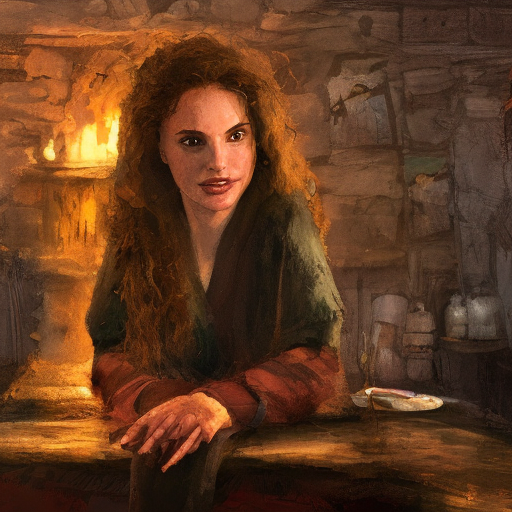}
        \caption{(8.0, 0.9370, 0.00063)}
    \end{subfigure}%
    \begin{subfigure}[b]{0.2\linewidth}
        \centering
        \includegraphics[height=1.2in]{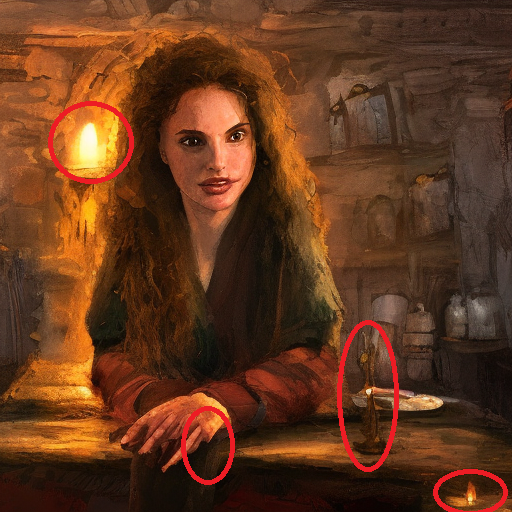}
        \caption{(10.0, 0.8741, 0.00069)}
    \end{subfigure}%
    
    \caption{ The original watermarked image, followed by images generated using our stealthy attack. Each attacked image is annotated with three values: the $\ell_2$ norm of the perturbation, Root Mean Squared Error (RMSE), and Structural Similarity Index Measure (SSIM), all measured relative to the original image.  In (e), we highlight noticeable artifacts introduced by the attack. The prompt used is ``{\em young, curly haired, redhead Natalie Portman  as a optimistic, cheerful, giddy medieval innkeeper in a dark medieval inn. dark shadows, colorful, candle light,  law contrasts, fantasy concept art by Jakub Rozalski, Jan Matejko, and J.Dickenson}''}
    \label{fig:diff_ssim_mse}
\end{figure*}

\begin{figure*}[!htbp]
    \centering
    \begin{subfigure}[b]{0.5\linewidth}
        \centering
        \includegraphics[height=1.5in]{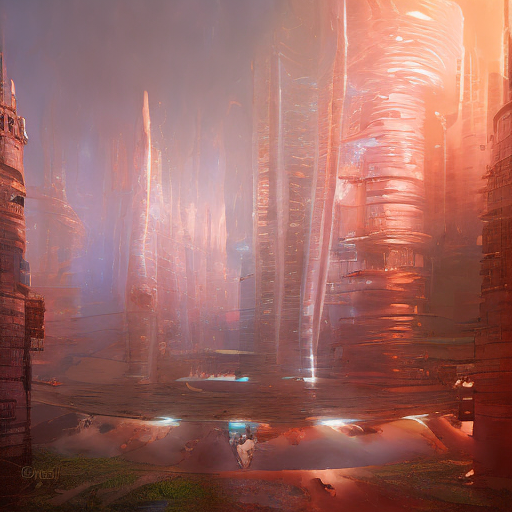}
        \includegraphics[height=1.5in]{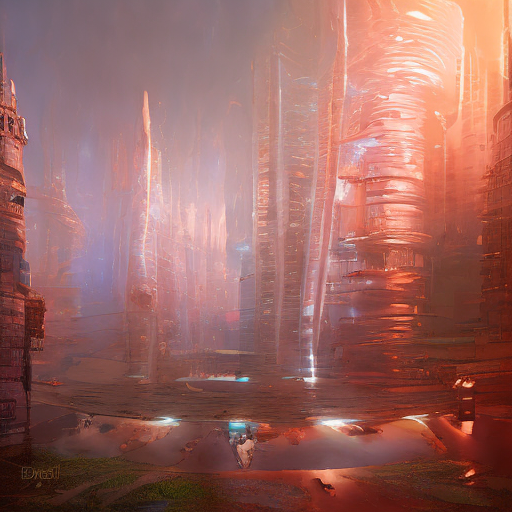}
        \caption{MSE $=0.0023$, SSIM $=0.9681$}
    \end{subfigure}%
    \begin{subfigure}[b]{0.5\linewidth}
        \centering
        \includegraphics[height=1.5in]{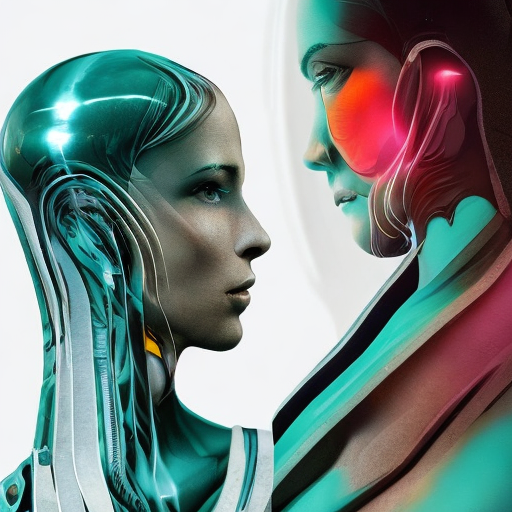}
        \includegraphics[height=1.5in]{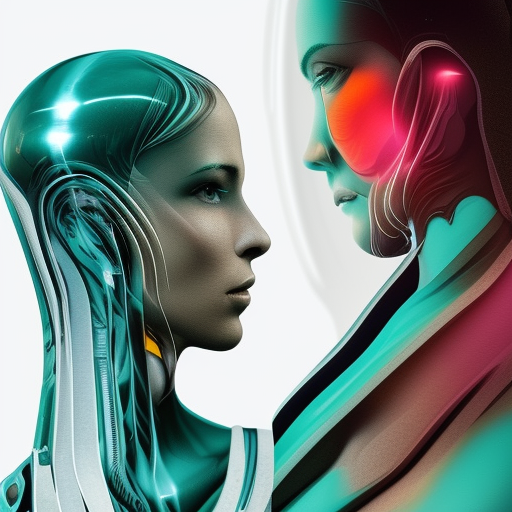}
        \caption{MSE $=0.0035$, SSIM $=0.9473$}
    \end{subfigure}%
    \\
    \begin{subfigure}[b]{0.5\linewidth}
        \centering
        \includegraphics[height=1.5in]{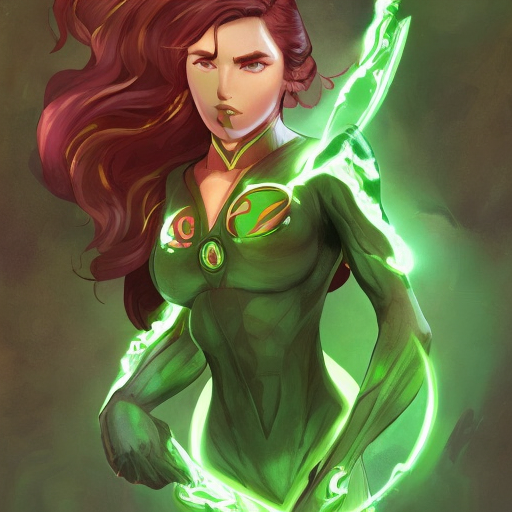}
        \includegraphics[height=1.5in]{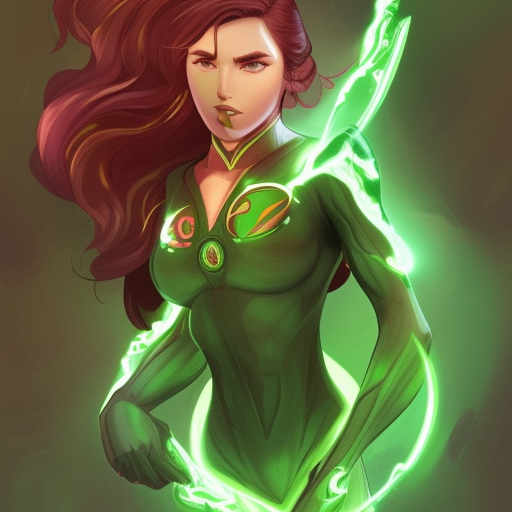}
        \caption{MSE $=0.0005$, SSIM $=0.9347$}
    \end{subfigure}%
    \begin{subfigure}[b]{0.5\linewidth}
        \centering
        \includegraphics[height=1.5in]{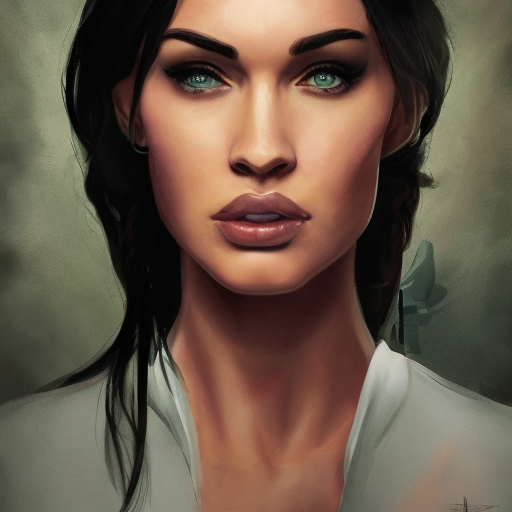}
        \includegraphics[height=1.5in]{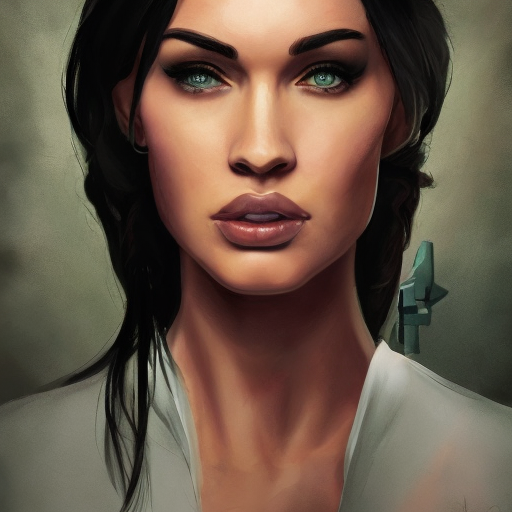}
        \caption{MSE $=0.0005$, SSIM $=0.9639$}
    \end{subfigure}%

    \caption{Examples of watermarked images and their corresponding watermark-removed versions. The removal perturbation has an $\ell_2$ norm of 8.0. Each subcaption shows the  RMSE and SSIM  between the original and the watermark-removed image.}
    \label{fig:example_removed}
\end{figure*}

\subsubsection{Boundary-Hiding Defense} \label{subsubsection:boundary_hiding_performance}
Consistent with our theoretical result described in Section \ref{section:defense}, when applying the boundary-hiding defense, the performance of the Stealthy attack (in AC3) reduces to that of the whitenoise attack. 



\subsection{Performance Overhead}
We evaluated the computational overhead introduced by our proposed boundary-hiding defense when integrated into the Stable Diffusion 2.1 pipeline. This defense mechanism introduces two additional computational steps into the initial latent generation phase and the watermark detection phase:

    

\subsubsection*{Time Overhead}
In total, the boundary-hiding defense introduces an additional computational overhead of approximately $1.1$\,s per image generated, excluding the one-time cost of generating the secret boundary-hiding transformation. This corresponds to an approximately $8.5\%$ increase in generation time.

\subsubsection*{Memory Overhead}
The defense introduces a memory overhead of approximately $1.07$\,GB,  due to the storage of the transformation matrix. We note that this matrix does not need to remain persistently on the GPU, as it is only used to (1) transform watermarked latents before the diffusion process and (2) invert reversed latents before watermark detection.

Therefore, the proposed boundary-hiding defense introduces a modest additional time and memory overhead, making it a practical and highly scalable defense.

\section{Discussion}
\subsection{Adversary with stronger capability} \label{discussion:known_watermark_boundary}
Known watermark attack assumes that the adversary can  query an oracle to obtain  random watermarked objects. However, this model offers limited adversarial capabilities.

A natural question arises: what constitutes a meaningful increase in attacker capability beyond the known watermark attack?  Consider a probabilistic  oracle $\widetilde{\mathcal{O}_{\epsilon}}$, which on query,  outputs a tuple $(s, s')=(s+\epsilon \cdot n, det(k,s') )$ , where $k$ is the secret key established during setup,  $s$ is a randomly generated watermark codeword using $k$, and  $n$ is whitenoise sampled from the multivariate Gaussian $\mathcal{N}(0,I)$. 
This enhanced oracle reveals more information, specifically the decision  boundary  of a randomly perturbed object.  Let us refer to such  setting as the {\em known boundary attack}.

While no adversary can have advantage over our proposed post-transformation method  under the {known watermark attack} (Theorem \ref{thm:mainresult}), we hypothesize that, given a sufficient number of queries to $\widetilde{\mathcal{O}_{\epsilon}}$,	 an adversary may be able to infer  information about the secret transformation. This raises the question of whether additional protective mechanism is possible. We believe this could be interesting future work.

\subsection{Scalability of Boundary Hiding Defense} \label{discussion:scalability}

The  proposed defense multiplies the $d$-dimensional starting latent point  with a randomly chosen orthonormal matrix.
 This matrix multiplication incurs a computational cost of $O(d^{2})$ operations.  While this is still manageable for models like  Stable Diffusion 2.1, where $d=4 \times 64 \times 64$, the computational cost grows quadratically with the number of dimensions. Such growth may become prohibitive in larger models, such as Stable Diffusion XL, where $d = 4 \times 128 \times 128$, or in use cases  requiring frequent transformations, for example video-generation.
 
To reduce the computational cost, it would be interesting to find classes of transformations that satisfy Condition \ref{condition:transformation} while allowing for efficient matrix-vector multiplication.  We believe that such transformations exist, see \cite{10.5555/3174304.3175339, matrix_lecture}. 
For example, it might be possible to perform sparse  transformation  in the Fourier domain and leverage on the Fast Fourier Transform (FFT) for acceleration. Other candidates may include transformations composed of sparse matrices that still satisfy the condition \cite{10.5555/3692070.3693765}.  Additionally, a relaxed version of  Condition \ref{condition:transformation} may suffice to  to hide the boundary.  We leave these directions for  future work.

\color{black}

\section{Related Works}

\subsection{Digital Watermarking}
Image watermarking has been a topic of extensive research for several decades, 
with the earliest known work on digital watermarking being published in 1994 \cite{413536}.
Although early work in digital watermarking of images is done in the pixel space (e.g. \cite{DBLP:conf/spieSR/KutterJB97, DBLP:conf/mm/KankanhalliRR98}), subsequent work primarily focus on watermarking the frequency domain of the digital object,
as such watermarking techniques are robust against mild affine and geometric perturbations. 
Since then, digital watermarking has been used for a variety of use-cases \cite{DBLP:journals/access/MalanowskaMAMK24}, such 
as copyright protection, content authentication, tamper detection, and fingerprinting. 

Recently, there has been growing interest in watermarking AI-generated content, such as the output of large language models (LLMs) and diffusion models. This may be due to several reasons, including tracking potential misuse of such models and embedding authorship information. \cite{DBLP:conf/icml/KirchenbauerGWK23} introduced the notion of undetectable watermarks in LLMs, which involves the use of the watermarking key as the pseudorandom number generator seed. Since then, there has been extensive work in ensuring that
such watermarked outputs are robust, i.e. the watermark in watermarked LLM outputs 
are still detectable even after mild paraphrasing. 

Watermarking techniques for diffusion models can be broadly categorized as follows. 
\begin{enumerate}
    \item \textbf{Post-processing watermarking}: Embeds watermarks directly into generated images but is vulnerable to common attacks like removal or ambiguity attacks\cite{DBLP:conf/ccnc/SongSML10}.
    \item \textbf{Fine-tuning-based watermarking}: Modifies diffusion models to embed watermarks during generation and trains a detector to extract them \cite{DBLP:conf/iccv/FernandezCJDF23, DBLP:journals/corr/abs-2405-02696}. However, these schemes have been shown to be removable by further fine-tuning the diffusion model \cite{hu2024stablesignatureunstableremoving}.
    \item \textbf{Latent space watermarking}: Embeds watermarks in the latent space starting point \cite{gunn2025undetectable,DBLP:conf/cvpr/YangZCF0Y24,wen2023treeringwatermarksfingerprintsdiffusion}. However, some methods, like Tree-Ring Watermarking \cite{wen2023treeringwatermarksfingerprintsdiffusion}, alter the latent distribution, making them detectable via steganographic attacks \cite{yang2024can}. Empirical evaluations \cite{gunn2025undetectable} also show that Gaussian Shading \cite{DBLP:conf/cvpr/YangZCF0Y24} introduces visually perceptible artifacts in the generated images. To the best of our knowledge, the PRC watermarking scheme \cite{gunn2025undetectable} is the only provably undetectable watermarking method and hence, is the focus of this paper. 
\end{enumerate}

\subsection{Watermark Removal Attacks}

Recently, there has been increasing interest in attacks targeting the removal of watermarks from the outputs of machine learning models \cite{DBLP:journals/corr/abs-2108-04974}. These attacks, known as watermark removal attacks, aim to perturb or alter watermarked outputs in a way that renders the embedded watermark undetectable while preserving the usability of the generated content \cite{DBLP:journals/corr/abs-2402-16187}. For instance, \cite{DBLP:conf/nips/ZhaoZSVGKVWL24} demonstrated that traditional pixel-level digital watermarking schemes are vulnerable to regeneration attacks, where a specially trained generative AI model creates a semantically similar but unwatermarked version of the image.

To address the limitations of post-hoc pixel-level watermarking, recent research has explored embedding watermarks into the initial latent starting point of the generative process \cite{DBLP:conf/cvpr/YangZCF0Y24, gunn2025undetectable, wen2023treeringwatermarksfingerprintsdiffusion}. By integrating the watermark at the latent level, the resulting scheme becomes more robust to common attacks such as noise addition, filtering, and geometric distortions, which typically compromise traditional pixel-level watermarking techniques. To the best of our knowledge, no prior work has specifically addressed the removal of watermarks embedded in the latent space.

\section{Conclusion}

In this paper, we present a watermark removal attack targeting watermarks in the latent starting point, demonstrating its effectiveness against the PRC watermarking scheme. This attack illustrates that undetectability does not guarantee robustness against removal. To counteract this, we propose a boundary-hiding defense mechanism that involves performing a secret orthonormal transformation on the watermarked latents.  We showed  that by performing the transformation on pseudorandom watermarking scheme which meets a ``well-behaved'' property,  no adversary can outperform the baseline whitenoise attacker. Empirical evaluations on Stable Diffusion 2.1 show that our attack successfully removes watermarks while preserving both image fidelity and semantics. Moreover, our defense provides a practical, scalable solution to mitigate such threats. These findings highlight the importance of hiding boundaries in latent-space watermarking schemes to ensure robustness against removal attacks. 

\begin{acks}
This research/project is supported by the National Research Foundation, Singapore under its AI Singapore Programme (AISG Award No: AISG3-RP-2022-029). We also thank the anonymous reviewers for their constructive feedback, and Wenjie Qu for noticing a citation mistake.
\end{acks}


\bibliographystyle{ACM-Reference-Format}
\bibliography{ref}

\appendix   
\section{Code Implementation} \label{appendix:code}
The code to reproduce our experiments together with instructions  is available in the following repository:
\url{https://github.com/dezhanglee/watermark_removal}. This code largely reuses the code base from the original PRC watermarking paper \cite{gunn2025undetectable}, which is distributed under the MIT License. 

\section{Samples of Attacked Images} \label{appendix:additional_img}
\begin{figure*}[]
    \centering
\begin{subfigure}[b]{0.32\linewidth}
    \centering
    \includegraphics[height=1.0in]{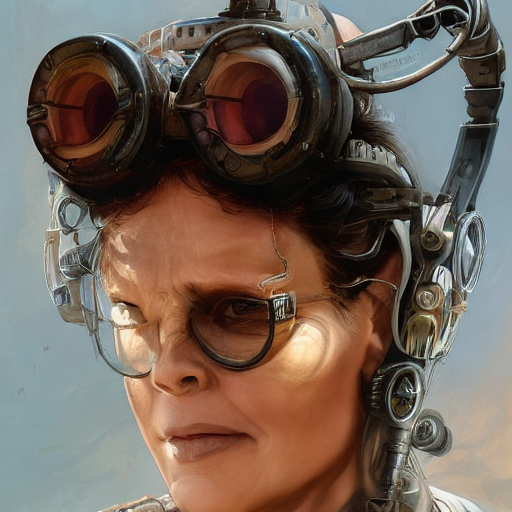}
    \includegraphics[height=1.0in]{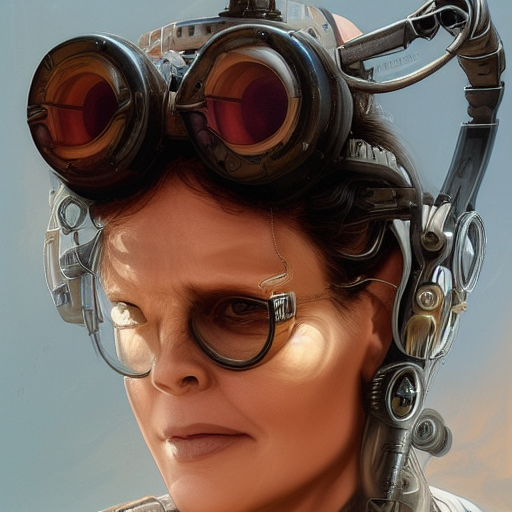}
    \caption{(0.0001, 0.9024)}
\end{subfigure}%
\begin{subfigure}[b]{0.32\linewidth}
    \centering
    \includegraphics[height=1.0in]{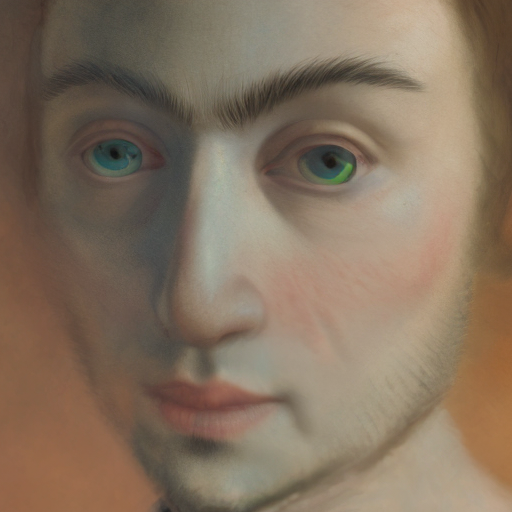}
    \includegraphics[height=1.0in]{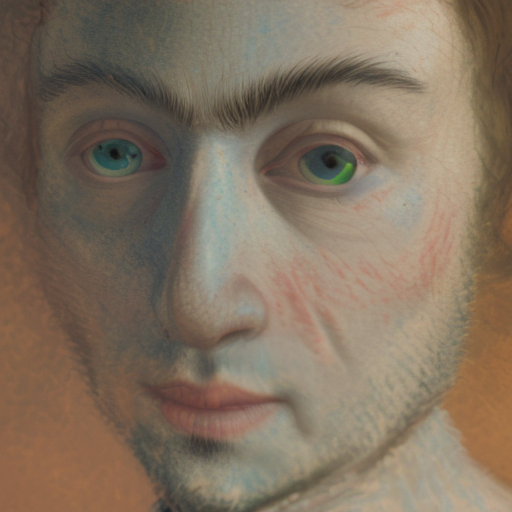}
    \caption{(0.0002, 0.8949)}
\end{subfigure}%
\begin{subfigure}[b]{0.32\linewidth}
    \centering
    \includegraphics[height=1.0in]{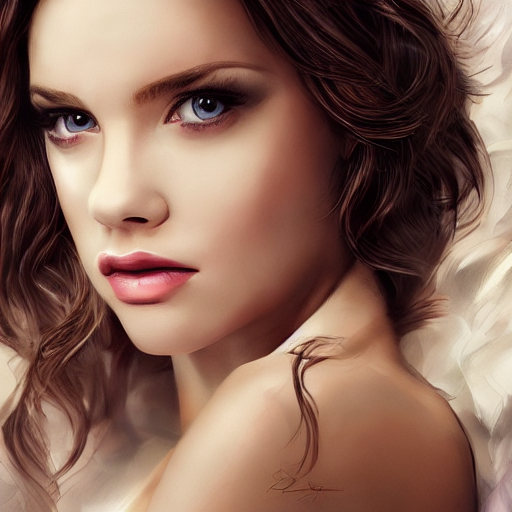}
    \includegraphics[height=1.0in]{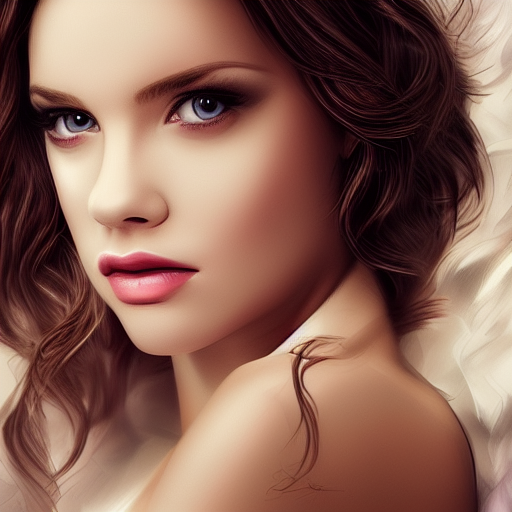}
    \caption{(0.0005, 0.9319)}
\end{subfigure}%
\\
\begin{subfigure}[b]{0.32\linewidth}
    \centering
    \includegraphics[height=1.0in]{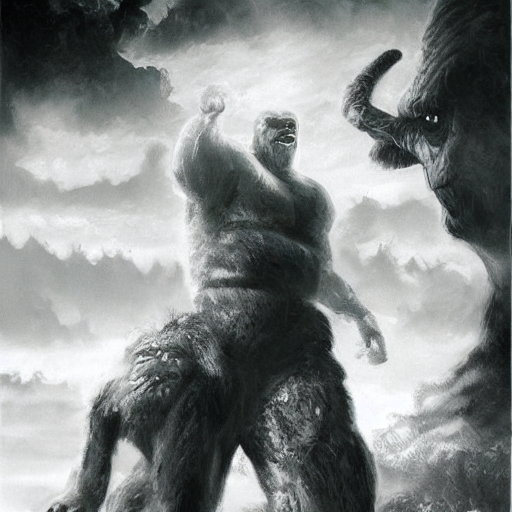}
    \includegraphics[height=1.0in]{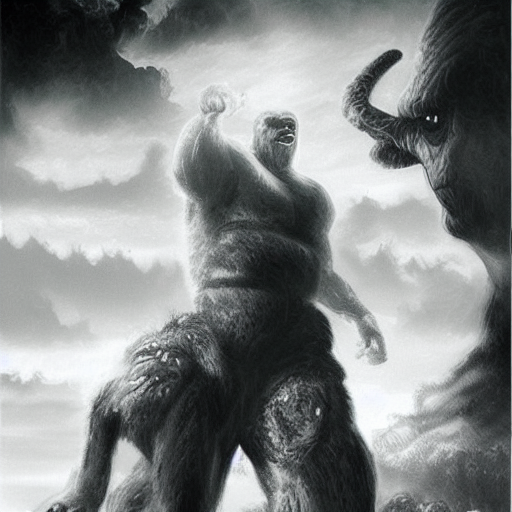}
    \caption{(0.0008, 0.9091)}
\end{subfigure}%
\begin{subfigure}[b]{0.32\linewidth}
    \centering
    \includegraphics[height=1.0in]{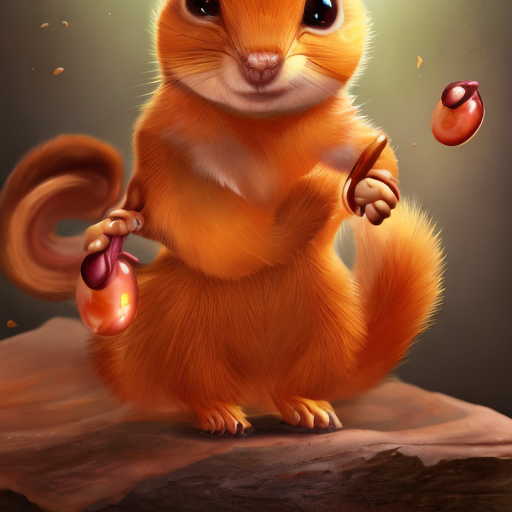}
    \includegraphics[height=1.0in]{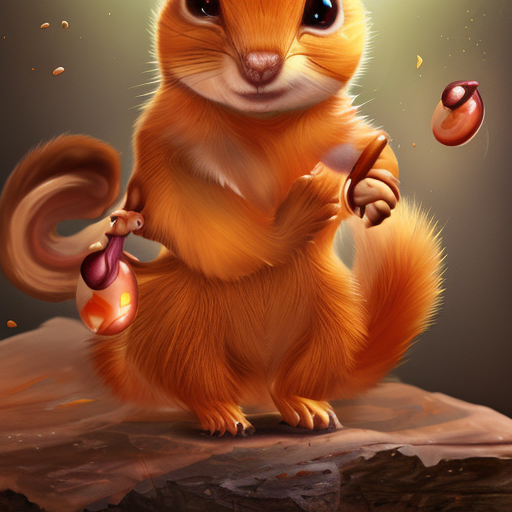}
    \caption{(0.0009, 0.9016)}
\end{subfigure}%
\begin{subfigure}[b]{0.32\linewidth}
    \centering
    \includegraphics[height=1.0in]{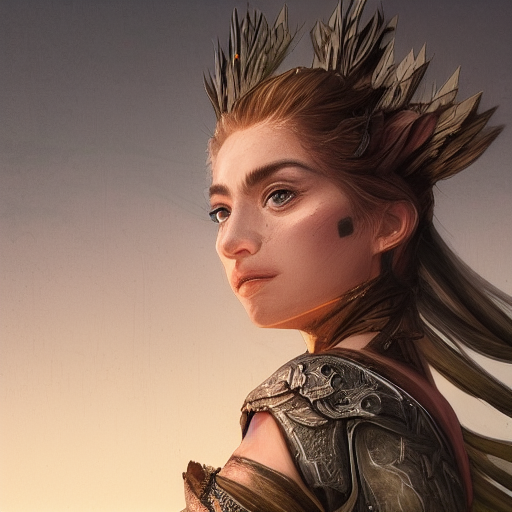}
    \includegraphics[height=1.0in]{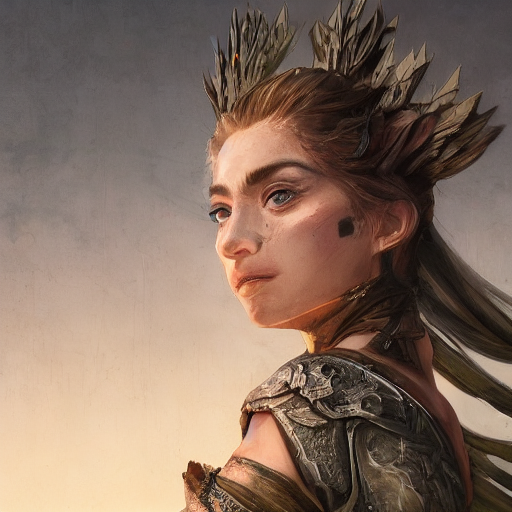}
    \caption{(0.0011, 0.8979)}
\end{subfigure}%
\\
\begin{subfigure}[b]{0.32\linewidth}
    \centering
    \includegraphics[height=1.0in]{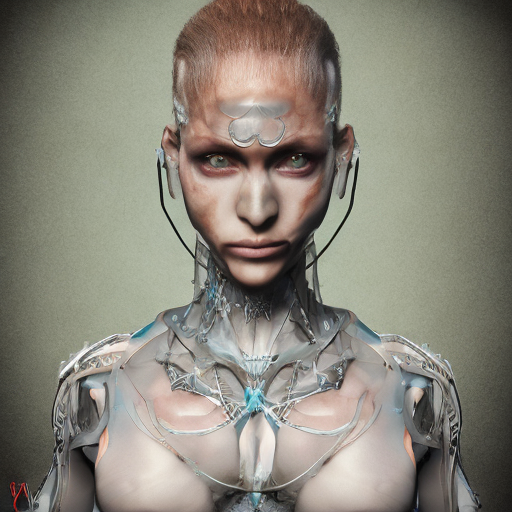}
    \includegraphics[height=1.0in]{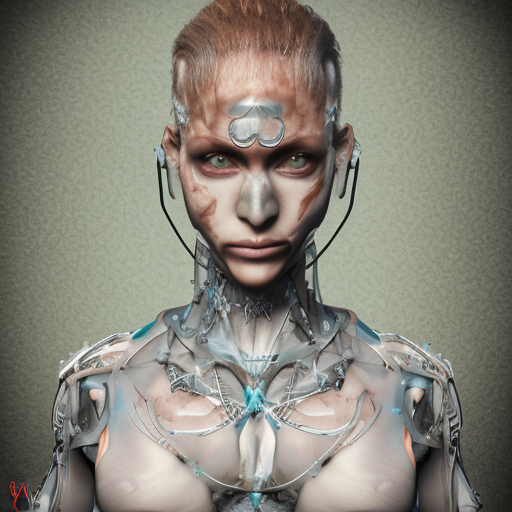}
    \caption{(0.0012, 0.8458)}
\end{subfigure}%
\begin{subfigure}[b]{0.32\linewidth}
    \centering
    \includegraphics[height=1.0in]{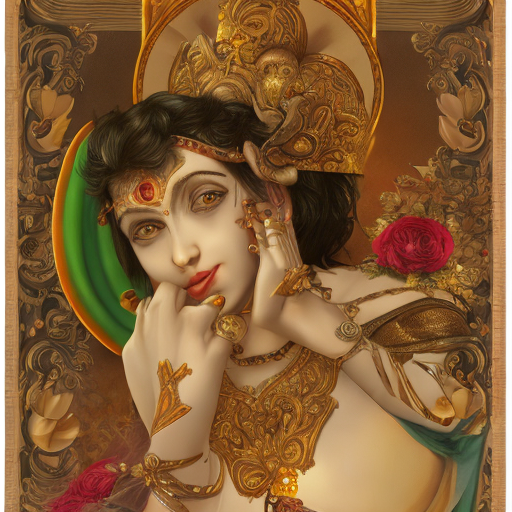}
    \includegraphics[height=1.0in]{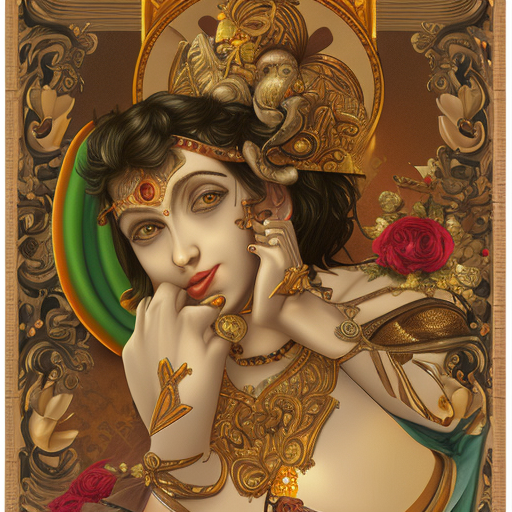}
    \caption{(0.0013, 0.7622)}
\end{subfigure}%
\begin{subfigure}[b]{0.32\linewidth}
    \centering
    \includegraphics[height=1.0in]{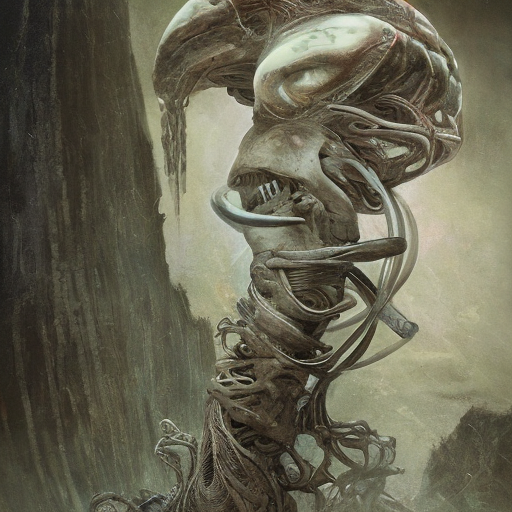}
    \includegraphics[height=1.0in]{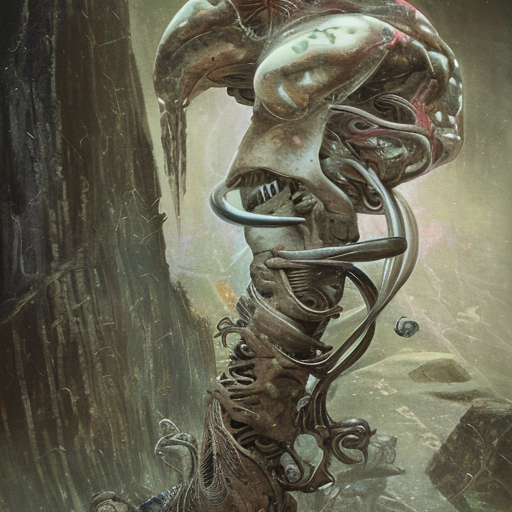}
    \caption{(0.0013, 0.7747)}
\end{subfigure}%
\\
\begin{subfigure}[b]{0.32\linewidth}
    \centering
    \includegraphics[height=1.0in]{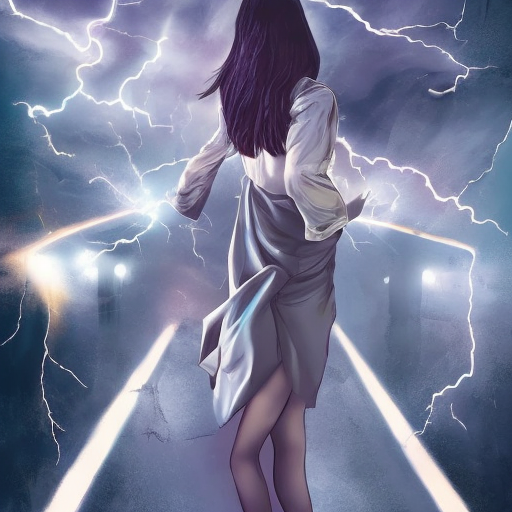}
    \includegraphics[height=1.0in]{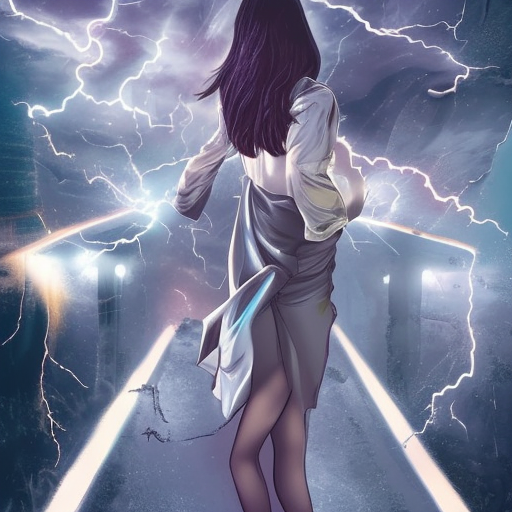}
    \caption{(0.0014, 0.8279)}
\end{subfigure}%
\begin{subfigure}[b]{0.32\linewidth}
    \centering
    \includegraphics[height=1.0in]{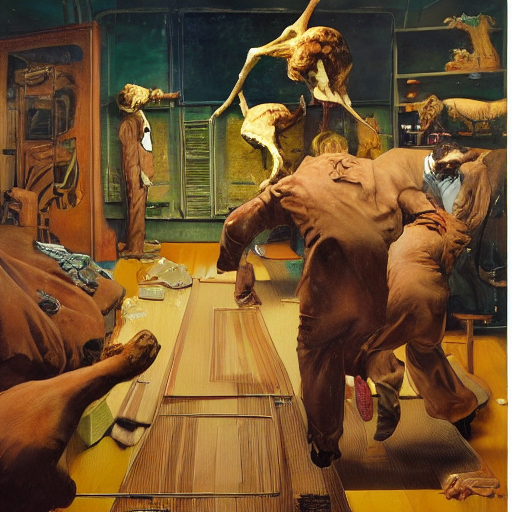}
    \includegraphics[height=1.0in]{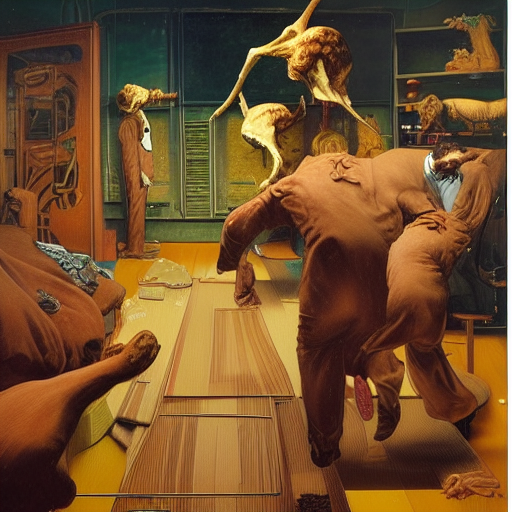}
    \caption{(0.0017, 0.8539)}
\end{subfigure}%
\begin{subfigure}[b]{0.32\linewidth}
    \centering
    \includegraphics[height=1.0in]{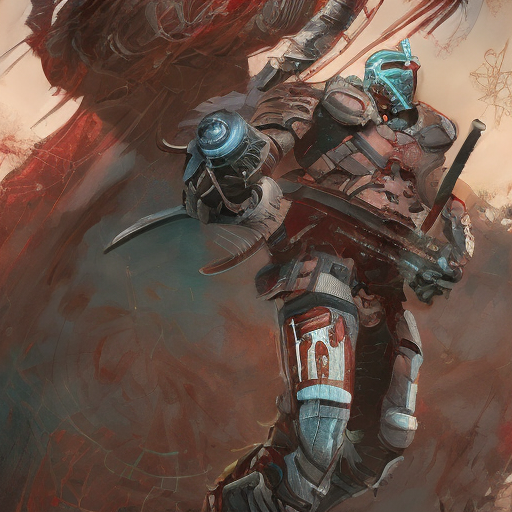}
    \includegraphics[height=1.0in]{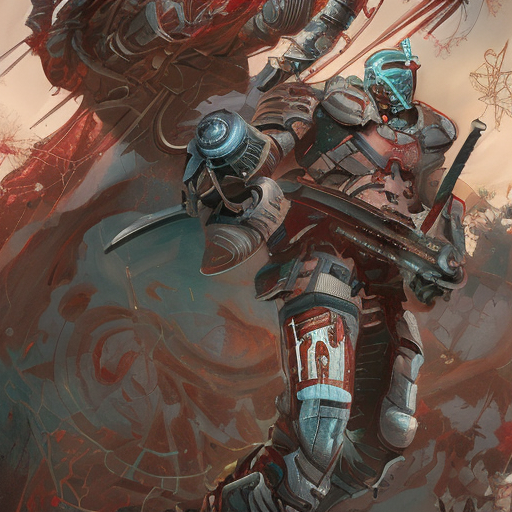}
    \caption{(0.0022, 0.8222)}
\end{subfigure}%
\\
\begin{subfigure}[b]{0.32\linewidth}
    \centering
    \includegraphics[height=1.0in]{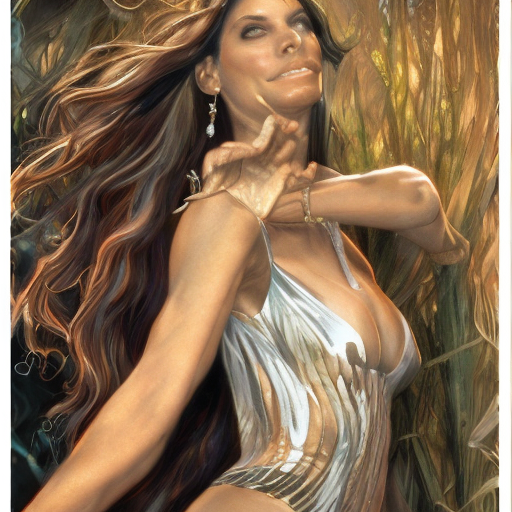}
    \includegraphics[height=1.0in]{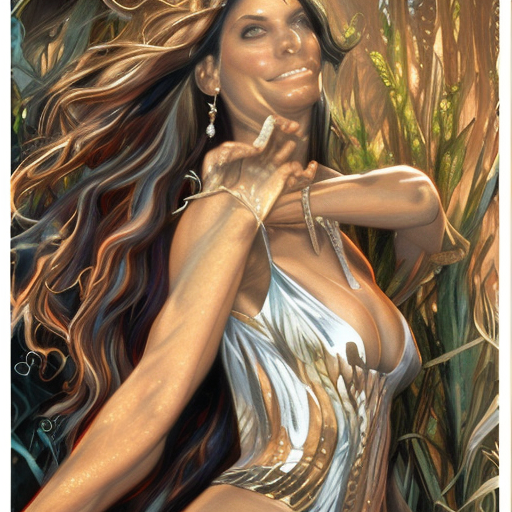}
    \caption{(0.0023, 0.7747)}
\end{subfigure}%
\begin{subfigure}[b]{0.32\linewidth}
    \centering
    \includegraphics[height=1.0in]{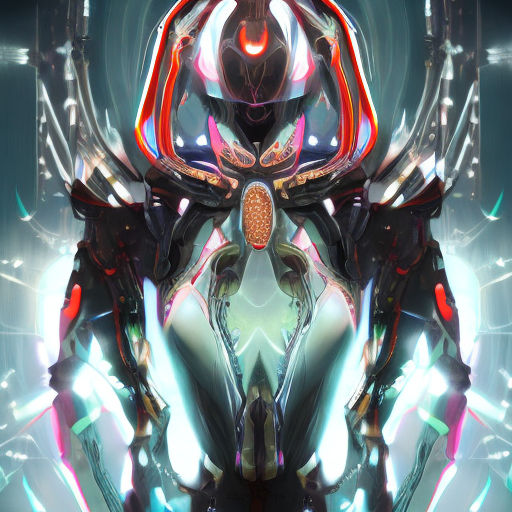}
    \includegraphics[height=1.0in]{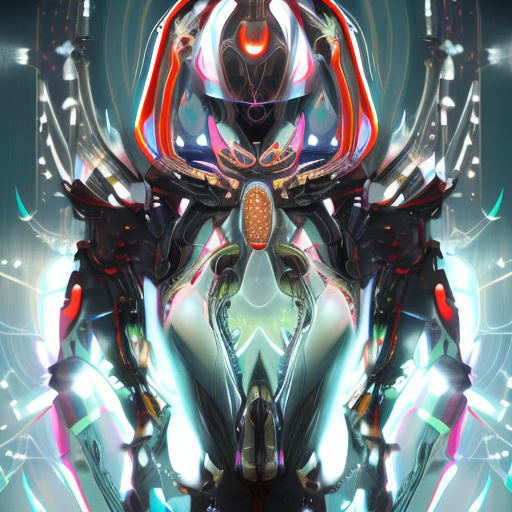}
    \caption{(0.0027, 0.8771)}
\end{subfigure}%
\begin{subfigure}[b]{0.32\linewidth}
    \centering
    \includegraphics[height=1.0in]{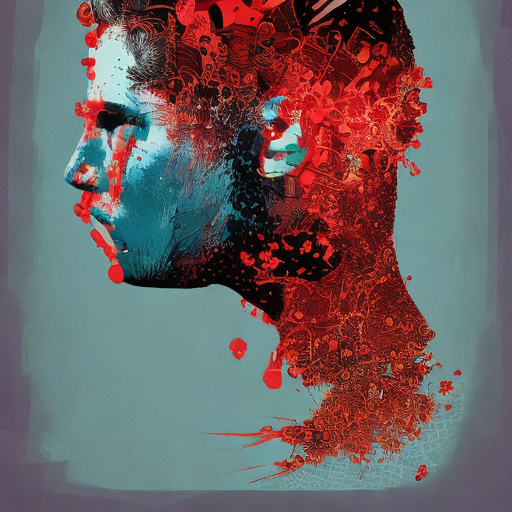}
    \includegraphics[height=1.0in]{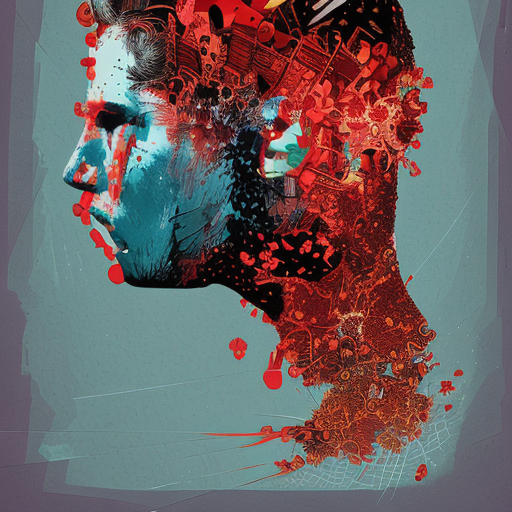}
    \caption{(0.0036, 0.7965)}
\end{subfigure}%
    \caption{Additional pairs of watermarked images (left), and their corresponding attacked images (right) generated using Stable Diffusion 2.1. Each subcaption shows the  RMSE and SSIM  between the image pairs, sorted by ascending RMSE.}
    \label{fig:example_additional_sd21}
\end{figure*}
Additional pairs of watermarked images and their corresponding versions after the stealthy attack are provided in Figure \ref{fig:example_additional_sd21}. Consistent with our evaluation results, these image pairs are perceptually similar. These images are part of a larger collection of $10,000$ image pairs, available at the following link: \url{https://github.com/dezhanglee/watermarked-images-samples}. They are generated using Stable Diffusion 2.1, under the same setting as described in Section \ref{section:eval_sd2.1}

\section{Performance Under Different LDM Models} \label{appendix:eval}
\subsection{Stable Diffusion 1.5}
In our evaluations, the parameters used are the same as the ones used for Stable Diffusion 2.1, as stated in Section \ref{section:eval_sd2.1}. The performance of our attacks and defense is similar to those reported in Section \ref{section:eval_sd2.1}, Table \ref{table:eval_sd_attack}. We believe that this is because both SD 2.1 and SD 1.5 shares the same latent space dimensions of $4 \times 64 \times 64$, and both models share largely the same architectures.

\begin{figure*}[]
    \centering
    \begin{subfigure}[b]{0.5\linewidth}
        \centering
        \includegraphics[height=1.5in]{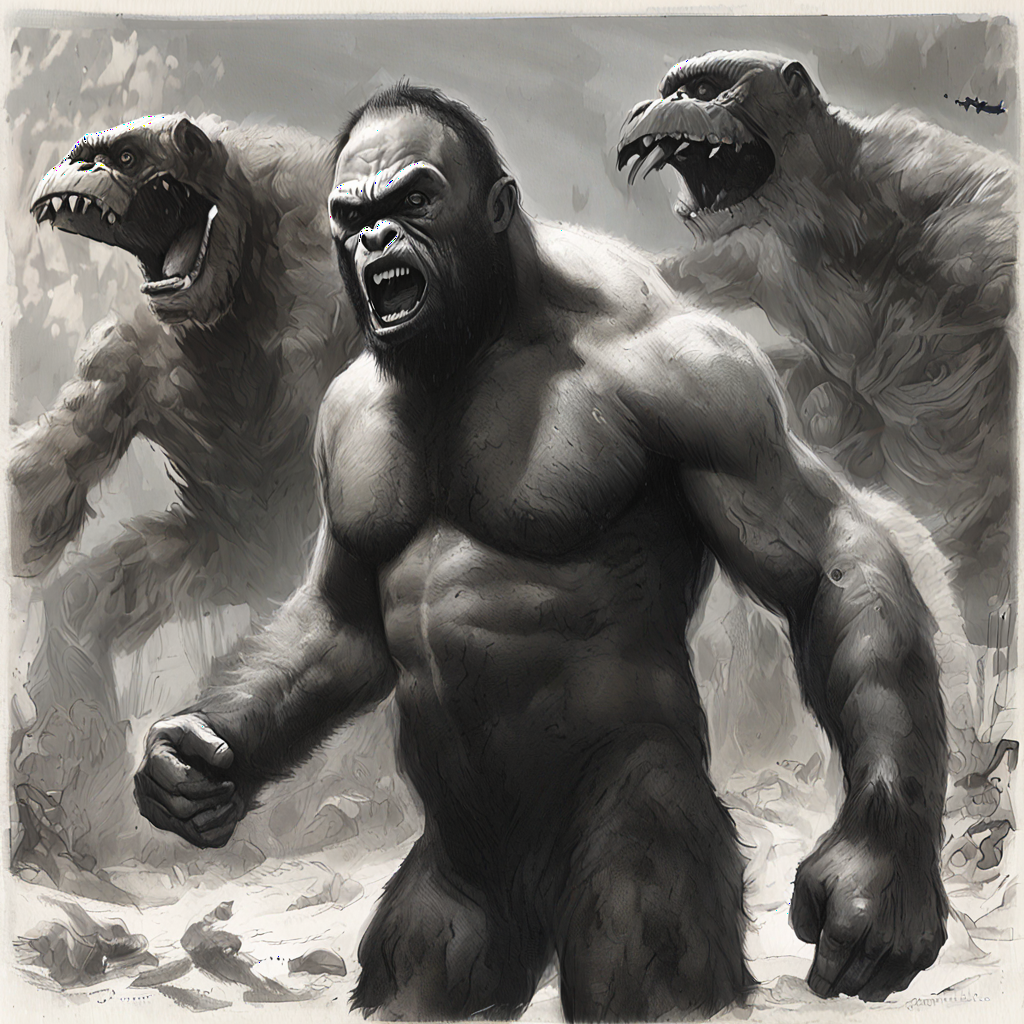}
        \includegraphics[height=1.5in]{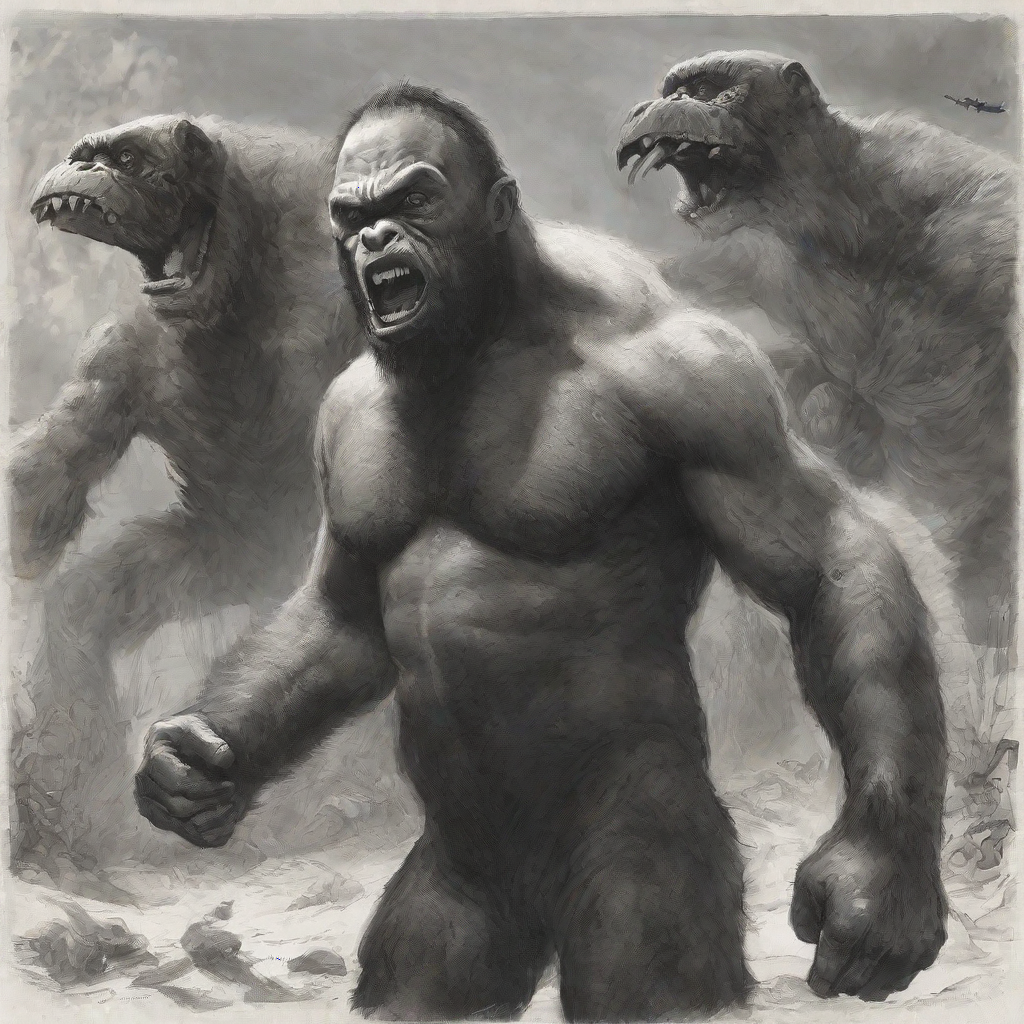}
        \caption{MSE $=0.0045$, SSIM $=0.7992$}
    \end{subfigure}%
    \begin{subfigure}[b]{0.5\linewidth}
        \centering
        \includegraphics[height=1.5in]{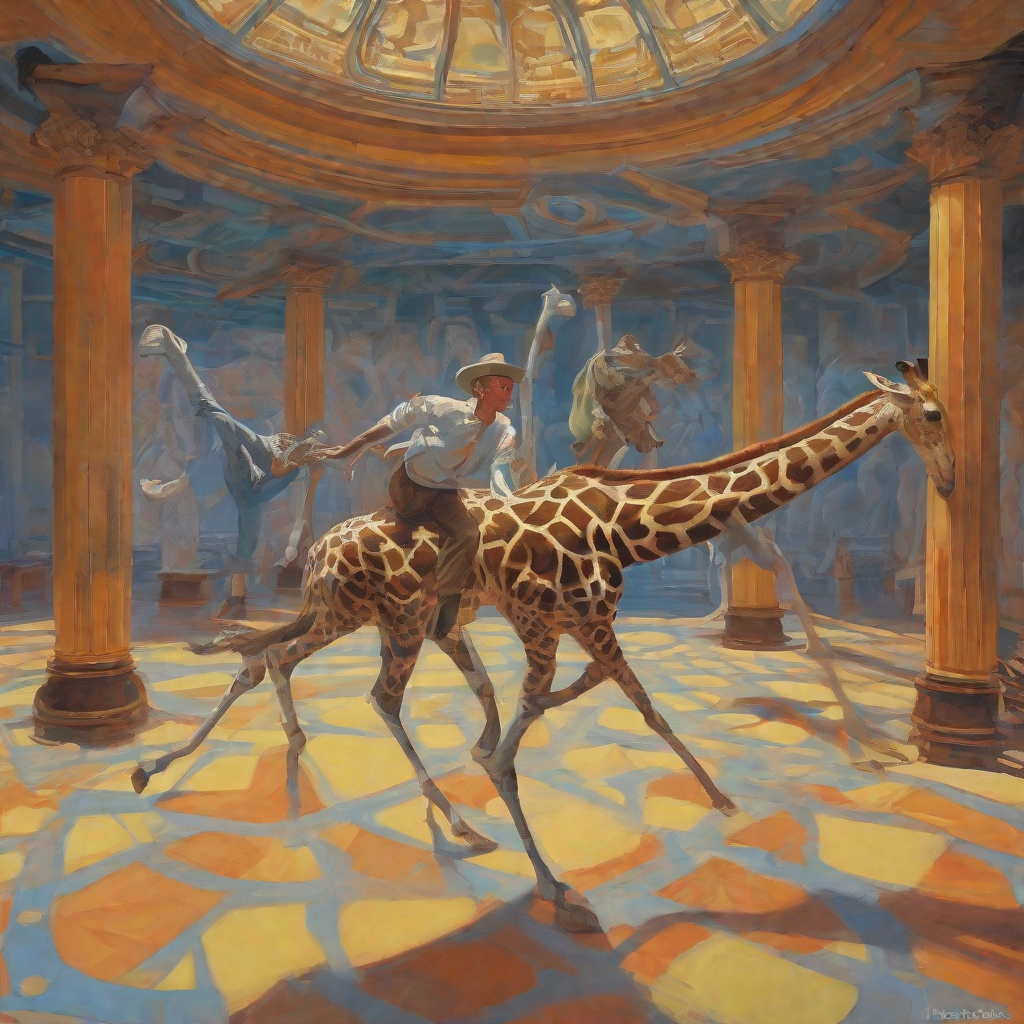}
        \includegraphics[height=1.5in]{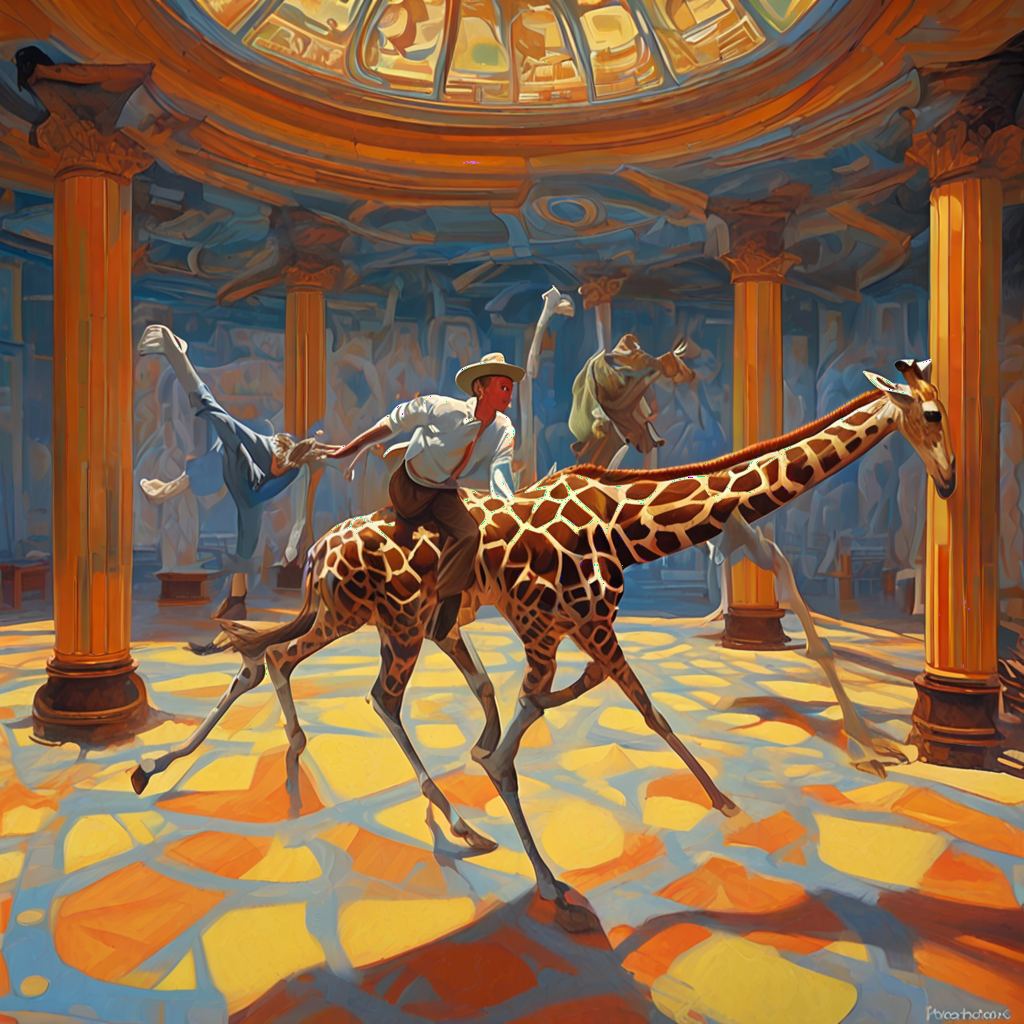}
        \caption{MSE $=0.0031$, SSIM $=0.7642$}
    \end{subfigure}%
    \\
     \begin{subfigure}[b]{0.5\linewidth}
        \centering
        \includegraphics[height=1.5in]{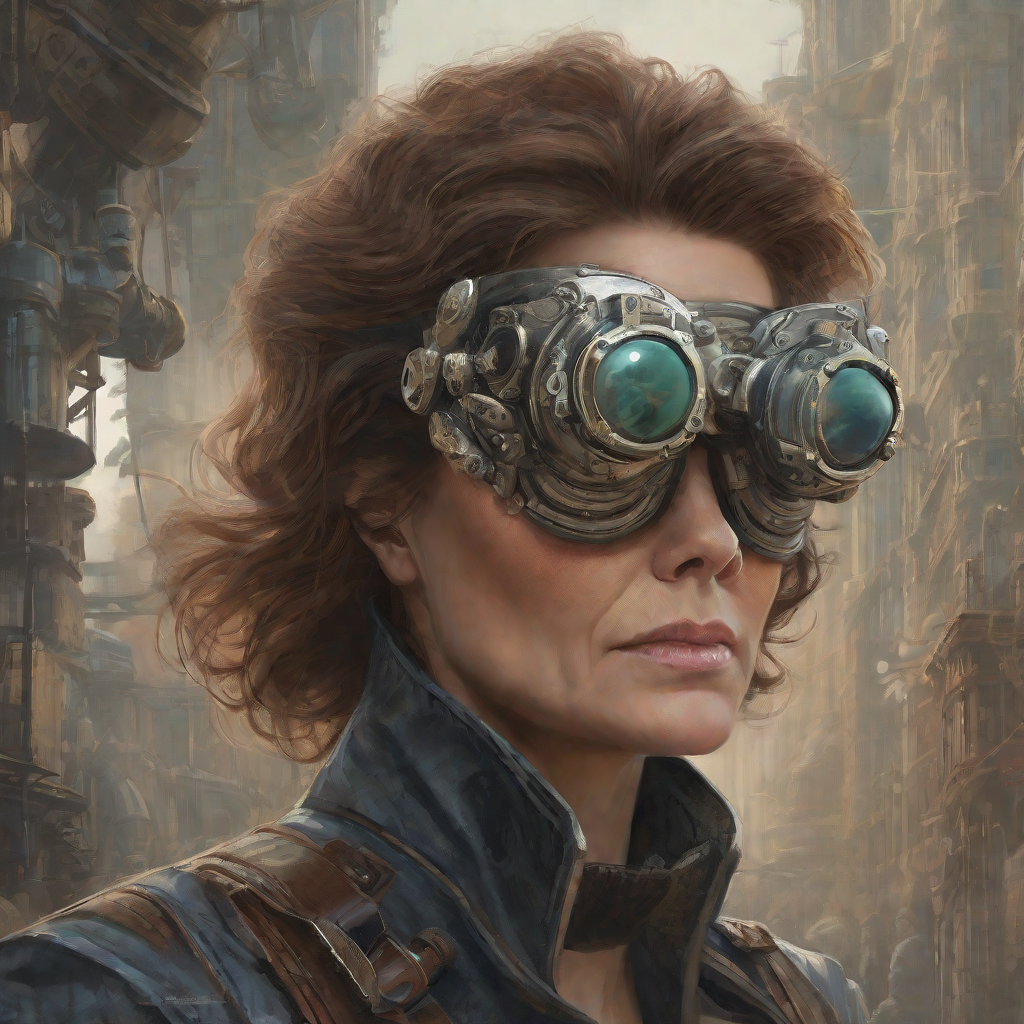}
        \includegraphics[height=1.5in]{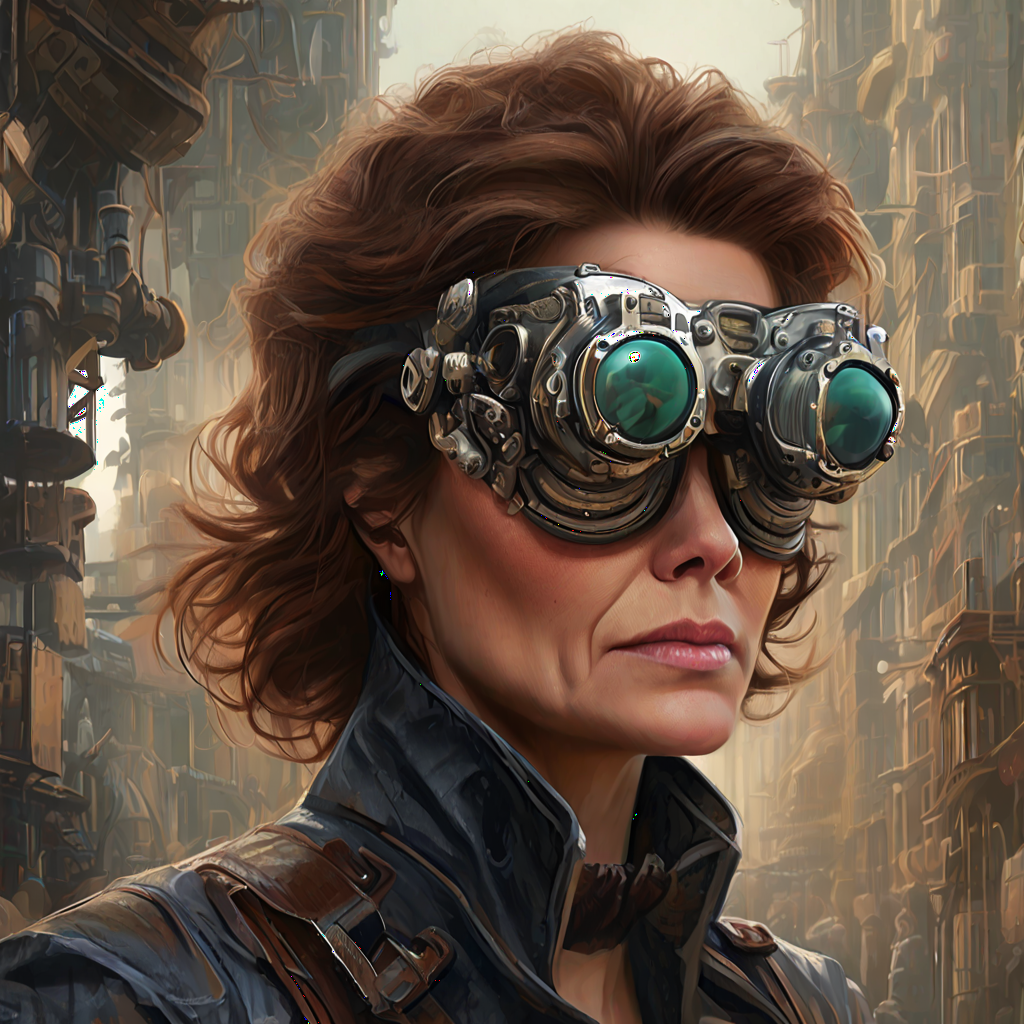}
        \caption{MSE $=0.0049$, SSIM $=0.7313$}
    \end{subfigure}%
    \begin{subfigure}[b]{0.5\linewidth}
        \centering
        \includegraphics[height=1.5in]{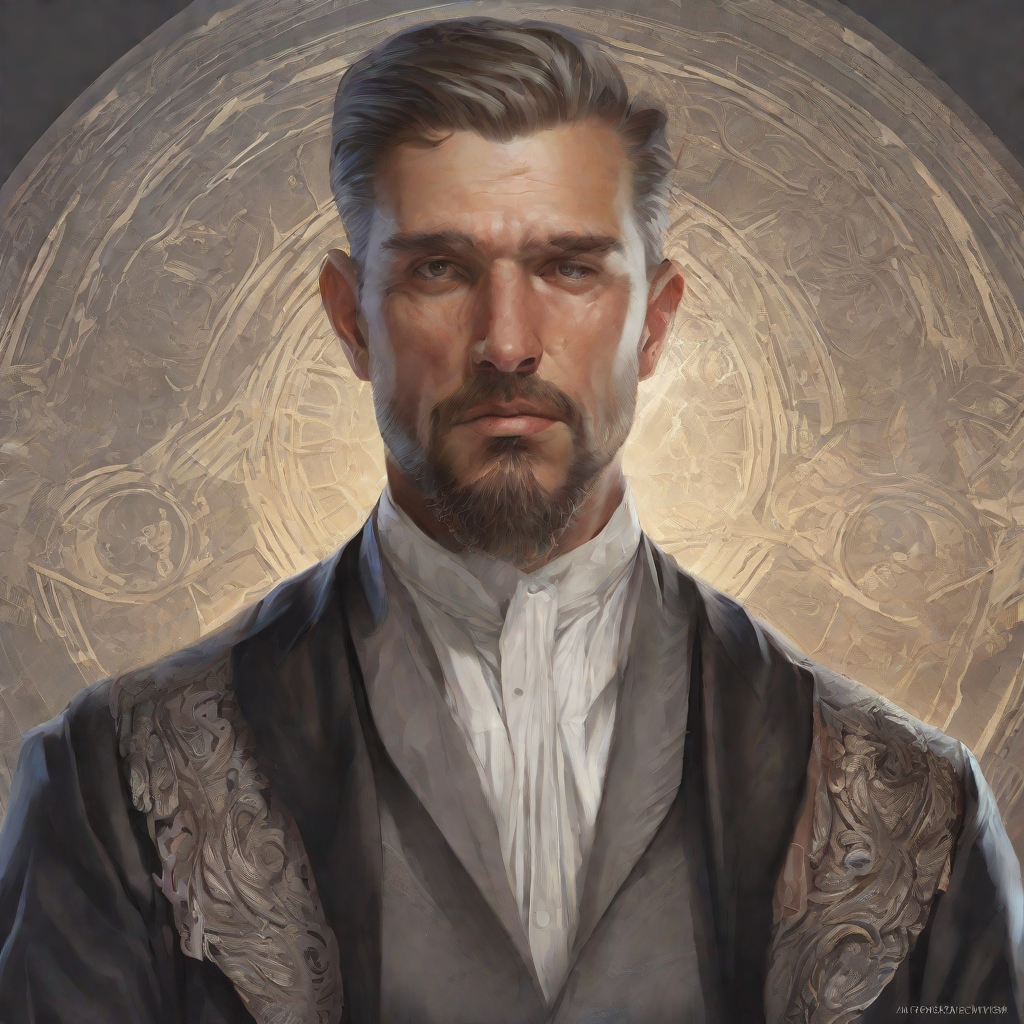}
        \includegraphics[height=1.5in]{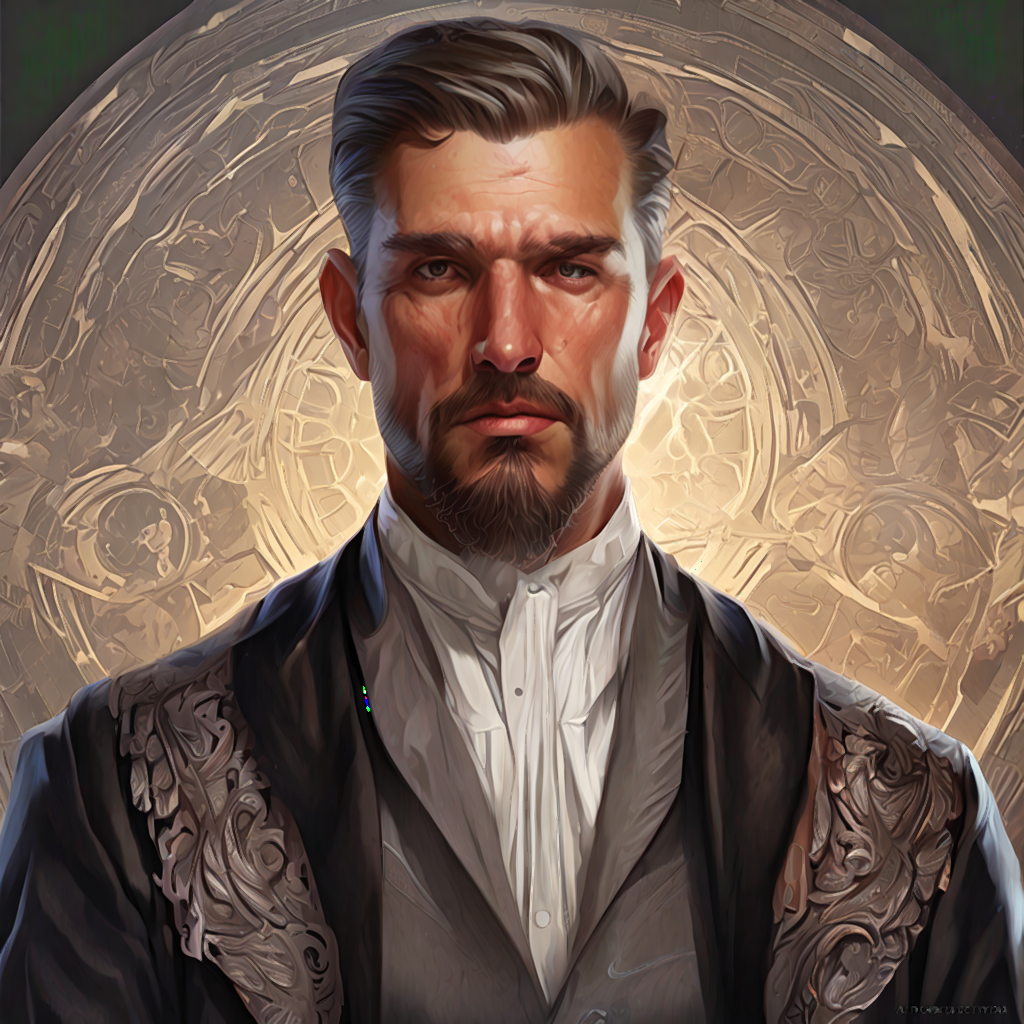}
        \caption{MSE $=0.0043$, SSIM $=0.7489$}
    \end{subfigure}%
    \\

    \caption{Examples of watermarked images (left) generated using SDXL, and their corresponding watermark-removed versions using the stealthy attack (right). The removal perturbation has an $\ell_2$ norm of 10.0. Each subcaption shows the  RMSE and SSIM  between the image pairs.}
    \label{fig:example_removed_sdxl}
\end{figure*}

\begin{table*}[]
\centering
\caption{Comparison of our proposed attack versus the white-noise attack under AC3, using Stable Diffusion XL. The results are consistent with our findings in Table \ref{fig:asr_attack_setting}. The result of the stealthy attack is {\ul \textbf{underlined in bold}}. }
\label{table:sdxl_asr}
\resizebox{2\columnwidth}{!}{%
\begin{tabular}{|l|cccccccccccrccc|}
\hline
\multirow{3}{*}{Attack} & \multicolumn{15}{c|}{Additional Perturbation, $\epsilon$ (Excluding Inversion Error)}                                                                                                                                                                                           \\ \cline{2-16} 
                        & \multicolumn{3}{c|}{4.0}                                 & \multicolumn{3}{c|}{6.0}                                 & \multicolumn{3}{c|}{8.0}                                 & \multicolumn{3}{c|}{10.0}                                & \multicolumn{3}{c|}{12.0}           \\ \cline{2-16} 
                        & ASR                 & RMSE   & \multicolumn{1}{c|}{SSIM} & ASR                 & RMSE   & \multicolumn{1}{c|}{SSIM} & ASR                 & RMSE   & \multicolumn{1}{c|}{SSIM} & ASR                 & RMSE   & \multicolumn{1}{c|}{SSIM} & ASR                 & RMSE   & SSIM \\ \hline
White Noise             & 0.01                & 0.0009 & \multicolumn{1}{c|}{0.82} & 0.03                & 0.0015 & \multicolumn{1}{c|}{0.82} & 0.03                & 0.0030 & \multicolumn{1}{c|}{0.80} & 0.09                & 0.0043 & \multicolumn{1}{r|}{0.78} & 0.12                & 0.0045 & 0.73 \\
{\ul \textbf{Stealthy}} & {\ul \textbf{0.06}} & 0.0009 & \multicolumn{1}{c|}{0.82} & {\ul \textbf{0.12}} & 0.0017 & \multicolumn{1}{c|}{0.81} & {\ul \textbf{0.19}} & 0.0033 & \multicolumn{1}{c|}{0.79} & {\ul \textbf{0.27}} & 0.0040 & \multicolumn{1}{r|}{0.75} & {\ul \textbf{0.33}} & 0.0049 & 0.73 \\ \hline
\end{tabular}%
}
\end{table*}


\subsection{Stable Diffusion XL}
We evaluated our attack and proposed defense on Stable Diffusion XL (SDXL) \cite{sdxl_paper}, using the implementation available on HuggingFace\footnote{\url{https://huggingface.co/stabilityai/stable-diffusion-xl-base-1.0}} under its default settings. Relative to SD 2.1, SDXL utilizes a larger latent space ($4 \times 128 \times 128$ compared to $4 \times 64 \times 64$), enabling the generation of higher-resolution images at $1024 \times 1024$ pixels (versus $512 \times 512$ for SD 2.1). 

Due to the $4\times$ larger latent space in SDXL, we increased the message length to $3,000$ bits. Given SDXL's significantly higher computational demands, our evaluation metrics were calculated based on 30 images generated from randomly selected prompts. We employed AC3, evaluating both the whitenoise and stealthy attack methods. 

\subsubsection{Evaluation Results}
The results of this evaluation are summarized in Table \ref{table:sdxl_asr}, and example pairs of attacked images are presented in Figure \ref{fig:example_removed_sdxl}. These results empirically demonstrate that the stealthy attack successfully removes watermarks while preserving image semantics.


\subsubsection{Performance of Boundary-Hiding Defense}
Consistent with our theoretical results in Section 5, when the boundary hiding defense is applied, the performance of the stealthy attack reduces to that of the whitenoise attack.

\subsubsection{Overhead of Boundary-Hiding Defense}
In our evaluations, the memory overhead of the secret transformation is $\sim 4.7$ GB, which is roughly $4\times $ the overhead incurred when the defense is employed on SD 2.1. This is consistent with our expectations, as the latent space of SDXL is $4 \times$ larger than SD 2.1.
\color{black}



\end{document}